\documentclass[11pt]{article}
\usepackage{amsfonts}
\usepackage{amsmath}
\usepackage{amssymb}
\usepackage{graphicx}
\usepackage{qcircuit}
\usepackage{subfig}
\usepackage{mathtools}
\usepackage[dvipsnames,svgnames]{xcolor}

\usepackage[colorlinks = true, linkcolor = blue, urlcolor=blue, citecolor = red, anchorcolor = green,]{hyperref}

\providecommand{\U}[1]{\protect\rule{.1in}{.1in}}

\newtheorem{theorem}{Theorem}

\newtheorem{algorithm}{Algorithm}

\newtheorem{lemma}{Lemma}

\newtheorem{proposition}{Proposition}
\newtheorem{remark}{Remark}

\newenvironment{proof}[1][Proof]{\noindent\textbf{#1.} }{\ \rule{0.5em}{0.5em}}
\DeclareMathOperator{\Tr}{Tr}

\allowdisplaybreaks

\begin{document}

\title{\vspace{-1.2in}\textbf{Efficient quantum algorithms for testing symmetries of open quantum systems}}
\author{Rahul Bandyopadhyay\thanks{Department of Electrical and Computer Engineering, University of California, Davis, California 95616, USA} \and
Alex H. Rubin\thanks{Department of Physics and Astronomy, University of California, Davis, California 95616, USA} \footnotemark[1] \and
Marina Radulaski\footnotemark[1] \and
Mark M.~Wilde\thanks{School of Electrical and Computer Engineering, Cornell University, Ithaca, New York 14850, USA}}
\maketitle

\begin{abstract}
Symmetry is an important and unifying notion in many areas of physics. In quantum mechanics, it is possible to eliminate degrees of freedom from a system by leveraging symmetry to identify the possible physical transitions. This allows us to simplify calculations and characterize potentially complicated dynamics of the system with relative ease. Previous works have focused on devising quantum algorithms to ascertain symmetries by means of fidelity-based symmetry measures. In our present work, we develop alternative symmetry testing quantum algorithms that are efficiently implementable on quantum computers. Our approach estimates asymmetry measures based on the Hilbert--Schmidt distance, which is significantly easier, in a computational sense, than using fidelity as a metric.  The method is derived to measure symmetries of states, channels,  Lindbladians, and measurements. We apply this method to a number of scenarios involving open quantum systems, including the amplitude damping channel and a spin chain, and we test for symmetries within and outside the finite symmetry group of the Hamiltonian and Lindblad operators.
\end{abstract}

\begin{quote}
\textit{We dedicate our paper to the memory of G\"oran Lindblad (July~9, 1940--November~30, 2022), whose profound contributions to quantum information science, in the form of the Lindblad master equation~\cite{Lindblad1976OnSemigroups} and the data-processing inequality for quantum relative entropy \cite{Lin75}, will never be forgotten.}
\end{quote}

\newpage

\tableofcontents

\section{Introduction}

Symmetry is a fundamental concept in physics, simplifying our understanding of
the physical world \cite{FR96,Gross96}. In quantum mechanics especially, symmetry is helpful for determining which physical transitions are allowed \cite{Wick1952,PhysRev.155.1428,BRS07} or in
reducing the number of degrees of freedom needed to express a given physical
system, thus making it easier to solve equations or optimization problems. In practical considerations, the interaction of the system with the environment can lead to a loss of symmetry, or yet, enforce certain symmetries (Figure \ref{fig:symmetry-concept-fig}). As
such, the concept of symmetry has carried over to quantum information
processing \cite{marvian2012symmetry}, for understanding phenomena like entanglement \cite{W89a,EW01,DPS02,DPS04,DPS05,CK07},
coherence \cite{LKJR15,MS16,SAP17}, and reference frames \cite{BRS07,GS08}. The essential role of
symmetry has elevated the concept itself to the status of a quantum resource
theory \cite{MS13,MS14}, in which objects possessing symmetry are considered freely available
and those that break symmetry have value. Most recently, symmetry is
being used in quantum machine learning to improve the trainability of learning
algorithms \cite{LSSVCC22,MMGMAWE23,Skolik2023}.

\begin{figure}
    \centering
    \includegraphics[width=\textwidth]{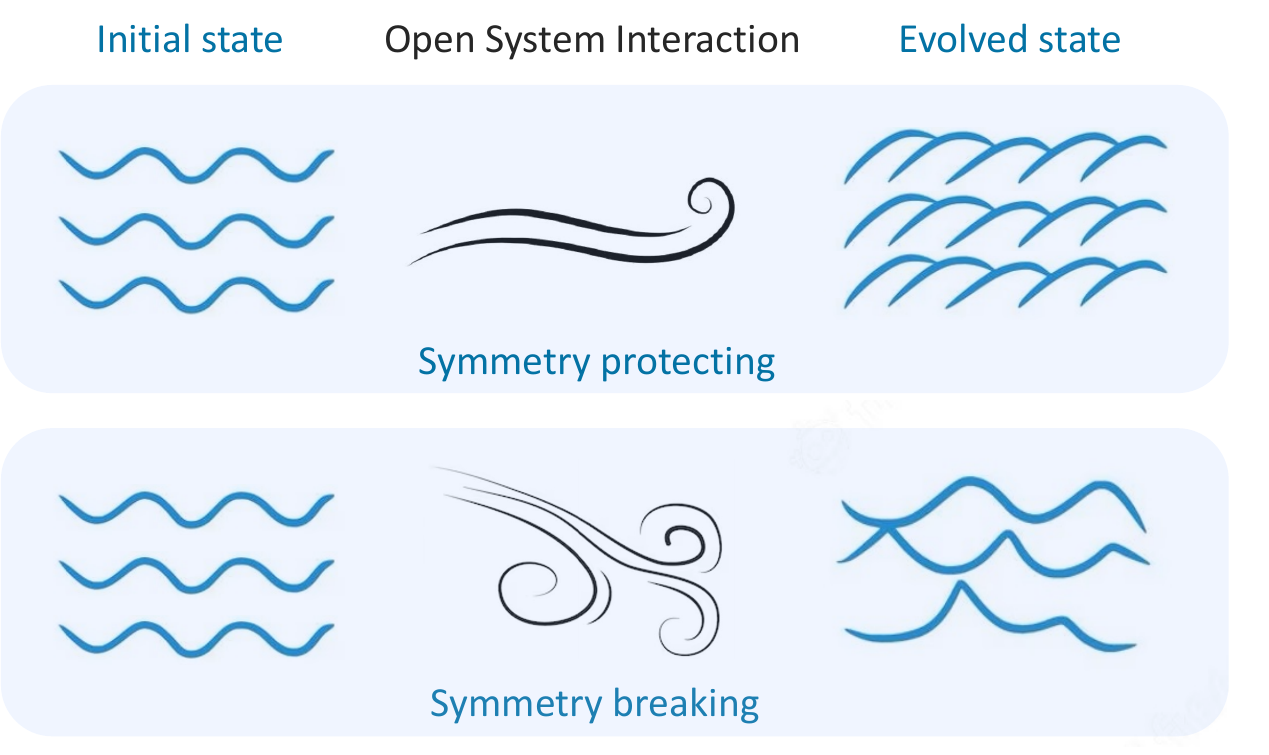}
    \caption{\small Interactions of systems with the environment can be symmetry preserving or symmetry breaking. The figure depicts an illustrative example of water waves interacting with wind that blows along different directions, potentially preserving or breaking the initial symmetry. In the first example (top), the wind preserves the symmetric structure of the water waves, so that the wind acts as a covariant channel, while in the second example (bottom), the wind is too chaotic, breaks the symmetry, and thus does not act as a covariant channel.}
    \label{fig:symmetry-concept-fig}
\end{figure}

Motivated by its fundamental role  in physics and related fields,
the authors of \cite{LRW22,laborde2022hamiltonian} (cf.~\cite{laborde2023menagerie}) developed several quantum algorithms for testing
symmetry of states, Hamiltonians, channels, and measurements on quantum
computers, and a sequel paper places the related problems in the context of
quantum computational complexity theory \cite{LRW23}. A number of these
algorithms are efficiently realizable on quantum computers, while others 
have computational complexity provably beyond that of the standard
BQP complexity class and thus are believed to be difficult even for quantum computers to solve (here, BQP stands for bounded error quantum polynomial time; see \cite{W09,VW15} for reviews on quantum computational complexity theory). Another contribution of \cite{LRW22} was to develop variational quantum algorithms for these
more difficult problems, by replacing the operations of an unbounded ``prover'' with parameterized quantum circuits; this approach works well in certain instances but does not lead to
provable computational runtimes (see \cite{CABBEFMMYCC20,bharti2021noisy} for
reviews of variational quantum algorithms).

One of the main contributions of the present paper is to develop alternative
symmetry-testing algorithms that can be efficiently
implemented on quantum computers. In contrast to the prior approaches from \cite{LRW22,laborde2022hamiltonian},
we modify the measure being
estimated by a quantum computer. Whereas all of the algorithms from \cite{LRW22}
estimate symmetry measures based on fidelity \cite{U76}, here we develop
algorithms that estimate asymmetry measures based on the Hilbert--Schmidt
distance. Since estimating fidelity is considered to be a difficult problem
for a quantum computer (more precisely, complete for a complexity class called
quantum statistical zero knowledge~\cite{W02}), while estimating the
Hilbert--Schmidt distance is considered easy for a quantum computer (more
precisely, complete for BQP \cite{RASW23}), it is expected that several of the symmetry testing
algorithms from \cite{LRW22} are difficult for a quantum computer while the symmetry
testing algorithms developed here are easy for a quantum computer to execute.

In our paper, we develop efficient symmetry testing algorithms for a number of
scenarios involving open quantum systems. Specifically, our contributions
consist of the following:

\begin{enumerate}
\item Given a state $\rho$ and a unitary representation $\left\{
U(g)\right\}  _{g\in G}$ of a group $G$, our first algorithm estimates the
following asymmetry measure:
\begin{equation}
\frac{1}{\left\vert G\right\vert }\sum_{g\in G}\left\Vert \left[
U(g),\rho\right]  \right\Vert _{2}^{2},
\label{eq:state-sym-meas}
\end{equation}
where
\begin{equation}
\label{eq:HS-norm}
\left\Vert A \right\Vert_2 \coloneqq \sqrt{\operatorname{Tr}[A^\dag A]}    
\end{equation}
is the Hilbert--Schmidt norm of an operator~$A$.
This measure is a faithful asymmetry measure, in the sense that it is equal to
zero if and only if $\left[  U(g),\rho\right]  =0$ for all $g\in G$, the
latter being the defining condition for symmetry of the state $\rho$ with
respect to the representation $\left\{  U(g)\right\}  _{g\in G}$~\cite{BRS07,GS08,marvian2012symmetry}.

\item Given a quantum channel $\mathcal{N}$ and a unitary channel
representation $\left\{  \mathcal{U}(g)\right\}  _{g\in G}$ of a group $G$, where $\mathcal{U}(g)(\cdot) \coloneqq U(g) (\cdot)U(g)^\dag$,
our next algorithm estimates the following asymmetry measure:
\begin{equation}
\frac{1}{\left\vert G\right\vert }\sum_{g\in G}\left\Vert \left(
\operatorname{id}\otimes\left[  \mathcal{U}(g),\mathcal{N}\right]  \right)
(\Phi^d)\right\Vert _{2}^{2},
\label{eq:channel-sym-meas}
\end{equation}
where $\operatorname{id}$ denotes the identity superoperator, $\left[  \mathcal{U}(g),\mathcal{N}\right]  $ represents the superoperator commutator
 (see, e.g., \cite[Section~II-C]{BRS07}), defined for superoperators $\mathcal{A}$ and $\mathcal{B}$ as
 \begin{equation}
     [\mathcal{A} , \mathcal{B}] \coloneqq  \mathcal{A}\circ \mathcal{B} - \mathcal{B}\circ \mathcal{A},
     \label{eq:superop-comm-def}
 \end{equation}
 and
\begin{equation}
\Phi^d \coloneqq \frac{1}{d} \sum_{i,j}|i\rangle\!\langle j|\otimes|i\rangle\!\langle j|
\label{eq:max-ent-state-def}
\end{equation}
is the standard maximally entangled state of Schmidt rank $d$. Thus,
\begin{multline}
\left(  \operatorname{id}\otimes\left[  \mathcal{U}(g),\mathcal{N}\right]
\right)  (\Phi^d)=\\
\left(  \operatorname{id}\otimes\left(  \mathcal{U}
(g)\circ\mathcal{N}\right)  \right)  (\Phi^d)-\left(  \operatorname{id}
\otimes\left(  \mathcal{N}\circ\mathcal{U}(g)\right)  \right)  (\Phi^d).
\end{multline}
As we show later on, the measure in \eqref{eq:channel-sym-meas} is a faithful
asymmetry measure, in the sense that it is equal to zero if and only if
\begin{equation}
\left[  \mathcal{U}(g),\mathcal{N}\right]  =0  \qquad \forall g\in G,
\label{eq:channel-symm-def-1}
\end{equation}
or, equivalently, if and only if
\begin{equation}  
\mathcal{U}(g)\circ\mathcal{N}=\mathcal{N}\circ\mathcal{U}(g) \qquad \forall g\in
G .
\label{eq:channel-symm-def-2}
\end{equation}
The latter is the defining condition for covariance symmetry of the channel~$\mathcal{N}$ with respect to the unitary channel representation $\{\mathcal{U}(g)\}_{g\in G}$ \cite{Holevo2002,marvian2012symmetry}.
In words, the
equality above means that the channel $\mathcal{N}$ commutes with every
unitary channel representation $\mathcal{U}(g)$ of a group element~$g \in G$.
Our algorithm for this task builds on an efficient subroutine for
estimating the Hilbert--Schmidt distance of the Choi states of two quantum
channels, which may be of independent interest for other purposes in quantum computing.

\item As a special case of the above, we consider testing covariance symmetry of measurement channels, which have the form $\rho \to \mathcal{M}(\rho) \coloneqq \sum_x \operatorname{Tr}[M_x \rho] |x\rangle\!\langle x|$, where $\{M_x\}_x $ is a positive operator-valued measure and $\{|x\rangle \}_x$ is an orthonormal basis that encodes the measurement outcome. Specifically, we provide an algorithm that estimates the following asymmetry measure:
\begin{equation}
    \frac{1}{|G|} \sum_{g \in G} \left \Vert \Phi^{\mathcal{M}
\circ\mathcal{U}(g)}-\Phi^{\mathcal{W}(g)\circ\mathcal{M}}\right \Vert_2^2,
\end{equation}
where $\{\mathcal{U}(g)\}_{g\in G}$ and $\{\mathcal{W}(g)\}_{g\in G}$ are unitary channel representations of a group $G$, with the latter realizing a shift of the measurement outcome as
\begin{equation}
\mathcal{W}(g)(|x\rangle\!\langle x|)= |\pi_g(x)\rangle\!\langle \pi_g(x)|,
\end{equation}
for $\pi_g$ a permutation. As discussed later on, this asymmetry measure is equal to zero if and only if the measurement is covariant \cite{D78,H11book}, i.e., such that $\mathcal{U}(g)(M_x)$ is an element of the POVM for all $g \in G$. Here again our algorithm builds on an efficient subroutine for estimating the Hilbert--Schmidt distance between two measurement channels, which we show is easier to perform than the aforementioned subroutine for general channels with quantum inputs and quantum outputs. We also believe that this subroutine should be of independent interest for other purposes in quantum computing.
\end{enumerate}

As a particular application of our algorithm for estimating
\eqref{eq:channel-sym-meas}, we investigate the symmetry of Lindbladian
evolutions, i.e., evolutions that correspond to the solution of the well known
Lindblad master equation \cite{Lindblad1976OnSemigroups}:
\begin{equation}
\frac{\partial\rho}{\partial t}=\mathcal{L}(\rho):=-i\left[  H,\rho\right]
+\sum_{k}L_{k}\rho L_{k}^{\dag}-\frac{1}{2}\{  L_{k}^{\dag}L_{k}
,\rho\}  ,\label{eq:lindblad-master-equation}
\end{equation}
where $H$ is a Hamiltonian, $\left\{  L_{k}\right\}  _{k}$ is a set of
Lindblad operators, and $\mathcal{L}$ is a superoperator known as the
Lindbladian. It is well known that the solution of
\eqref{eq:lindblad-master-equation} is the following quantum channel:
\begin{equation}
e^{\mathcal{L}t}(\rho)=\sum_{n=0}^{\infty}\frac{\mathcal{L}^{n}(\rho
) t^{n}}{n!},
\end{equation}
where $\mathcal{L}^{n}$ denotes $n$ repeated applications of the superoperator
$\mathcal{L}$. We accomplish symmetry testing of a Lindbladian $\mathcal{L}$ by employing our algorithm for estimating \eqref{eq:channel-sym-meas} with the substitution $\mathcal{N} = e^{\mathcal{L} t}$, and later on, we remark on how symmetry testing of the channel $e^{\mathcal{L} t}$ is equivalent to symmetry testing of the Lindbladian $\mathcal{L} $.

Similar to how understanding symmetries of Hamiltonians can be helpful for
deducing which physical transitions are allowed and which are not, the same can be
said for understanding symmetries of the more general Lindbladian evolutions.
As a particular example of this phenomenon, consider a Lindbladian in which
the Hamiltonian is the photon number operator \cite{GK04} and there is one Lindblad
operator, which is also the photon number operator. Then the only states that
are invariant under the resulting channel~$e^{\mathcal{L}t}$ are the photon
number states and mixtures thereof, because every other state becomes dephased
by this evolution. Thus, under these dynamics and for long times, it is not possible to transition from a probabilistic mixture of photon number states to a coherent superposition of them, the latter of which is resourceful for estimation tasks in quantum metrology \cite{Toth_2014}. More generally, our algorithm is helpful for understanding
symmetries of Lindbladian evolutions that are efficiently realizable on
quantum computers, by means of any of the several quantum algorithms that have
been proposed for simulating open systems dynamics \cite{Childs2016EfficientDynamics,Cleve2016EfficientEvolution,KSMM22,Schlimgen2022QuantumOperators,Suri2022Two-UnitarySimulation} (see \cite{Miessen2022QuantumDynamics} for a review).

Before proceeding with the content of our paper, we note here that the symmetry testing quantum algorithms proposed here, like those from \cite{LRW22,laborde2022hamiltonian}, are most useful in the regime in which the states, channels, Lindbladians, or measurements being tested, as well as the group representation unitaries being considered, involve a large number of qubits and are non-trivial. In this regime, it is likely not possible to simulate these tests efficiently by means of a classical computer, as shown in \cite{laborde2022hamiltonian,LRW23}, based on the conjecture that the complexity class BQP strictly contains the complexity class BPP (the latter being the class of problems that are efficiently implementable on a classical probabilistic computer). The previous statement, less formally, is equivalent to the widespread belief that quantum computers, in principle, are generally more powerful than classical computers. Furthermore, it is certainly of interest to employ quantum computers for the task of learning symmetries (see, e.g., \cite{lu2023learning}), and we consider the ability to test symmetries to be an important component of the learning process (either while the learning is occurring or after learning has completed, as a way of testing whether the learned symmetry is indeed correct).

In the rest of our paper, we provide details of our algorithms and evaluate their performance for some exemplary physical systems of interest. In particular, Section~\ref{sec:notation} reviews some basic notation and concepts used throughout the rest of our paper. Section~\ref{sec:symm-testing-algos} develops the theory behind our quantum algorithms for testing symmetry of states (Section~\ref{sec:symm-testing-states}), channels (Section~\ref{sec:symm-tests-channels-algos}), and Lindbladians (Section~\ref{sec:symm-tests-lindbladians-algos}). As part of our algorithm for testing symmetries of channels, we develop an efficient subroutine for estimating the Hilbert--Schmidt distance of the Choi states of two quantum channels (Section~\ref{sec:HS-Choi-states-algo}), which may be of independent interest for other purposes in quantum computing. Specifically, this algorithm significantly reduces the number of qubits needed for the estimation, when compared to a naive approach to this problem. In Section~\ref{sec:sims}, we test out our algorithms for estimating symmetries of Lindbladians for two example scenarios, using Qiskit's noiseless and noisy simulators \cite{Qiskit}. Section~\ref{sec:meas-ch} particularizes the development for quantum channels to the case of quantum measurement channels, proposing both a procedure for estimating the Hilbert--Schmidt distance of the Choi states of two such channels, as well as for estimating an asymmetry measure for a given measurement channel. Finally, in Section~\ref{sec:conclusion}, we conclude with a summary of our contributions, along with a discussion of prospects for implementing the developed algorithms on near-term quantum hardware.

\section{Notation and background}

\label{sec:notation}

This section provides some notation and background used throughout the rest of our paper. See \cite{hayashi_2017-book,wilde_2017,watrous_2018,holevo_2019,KW20} for further background on quantum information. A quantum state (density operator) is described by a positive semi-definite operator with unit trace. A quantum channel is a completely positive, trace-preserving superoperator. The Choi state $\Phi^{\mathcal{N}}$ of a channel~$\mathcal{N}$ is given by sending one share of a maximally entangled state $\Phi^d$, defined in~\eqref{eq:max-ent-state-def}, through the channel:
\begin{equation}
    \Phi^{\mathcal{N}} \coloneqq (\operatorname{id} \otimes \mathcal{N})(\Phi^d),
    \label{eq:choi-state-def}
\end{equation}
where we have assumed that the input space of $\mathcal{N}$ is $d$-dimensional.

\subsection{Hilbert--Schmidt distance}

The Hilbert--Schmidt distance between two states $\rho$ and $\sigma$, induced by the norm in \eqref{eq:HS-norm}, is given by $\left\Vert \rho - \sigma \right\Vert_2$. It is faithful, in the sense that $\left\Vert \rho - \sigma \right\Vert_2 = 0$ if and only if $\rho = \sigma$. It obeys the data-processing inequality for unital channels \cite{PWPR06}, but it does not obey it in general \cite{Ozawa2000}; that is, the following inequality holds whenever $\mathcal{N}$ is a unital channel (satisfying $\mathcal{N}(I) = I$, where~$I$ is the identity operator):
\begin{equation}
    \left\Vert \rho - \sigma \right\Vert_2 \geq \left\Vert \mathcal{N}(\rho) - \mathcal{N}(\sigma) \right\Vert_2.
\end{equation}

When $\rho$ and $\sigma$ are multi-qubit states and one can prepare many copies of them on a quantum computer, it is easy to estimate the square of their Hilbert--Schmidt distance by means of the destructive SWAP test (reviewed in Section~\ref{sec:dest-SWAP-test-review} below). This follows by considering the expansion
\begin{equation}
    \left\Vert \rho - \sigma \right\Vert_2^2 = \Tr[\rho^2] + \Tr[\sigma^2] - 2\Tr[\rho \sigma],
    \label{eq:HS-expansion}
\end{equation}
and the algorithm reviewed in the next section.
In fact, it is known that estimating the Hilbert--Schmidt distance of quantum states $\rho$ and $\sigma$ prepared by circuits is a BQP-complete problem \cite[Theorem~14]{RASW23}, so that this problem captures and is equivalent to the full power of quantum computation.

\subsection{Review of destructive SWAP test}

\label{sec:dest-SWAP-test-review}

Let us define the unitary swap
operator as
\begin{equation}
\operatorname{SWAP}\coloneqq \sum_{i,j}|i\rangle\!\langle j|\otimes|j\rangle\!\langle
i|,
\end{equation}
and note the following identity:
\begin{equation}
\operatorname{Tr}[CD]=\operatorname{Tr}[\operatorname{SWAP}(C\otimes
D)],
\label{eq:swap-prod-id}
\end{equation}
which holds for arbitrary linear operators $C$ and $D$ and plays a key role in our algorithms that follow. Recall that, if the SWAP operator acts on qubit systems, then
\begin{equation}
\operatorname{SWAP}    =\sum_{i,j\in\left\{  0,1\right\}  }\left(  -1\right)
^{ij}\Phi^{ij},
\label{eq:swap-bell}
\end{equation}
where
\begin{equation}
\Phi^{00}=\Phi^{+},\quad\Phi^{10}=\Phi^{-},\quad\Phi^{01}=\Psi^{+},\quad
\Phi^{11}=\Psi^{-}.
\end{equation}
In the above, $\Phi^{+}\equiv |\Phi^{+}\rangle\!\langle \Phi^{+}|$, $\Phi^{-}\equiv |\Phi^{-}\rangle\!\langle \Phi^{-}|$, $\Psi^{+}\equiv |\Psi^{+}\rangle\!\langle \Psi^{+}|$, and $\Psi^{-}\equiv |\Psi^{-}\rangle\!\langle \Psi^{-}|$ are the standard Bell states, defined through 
\begin{equation}
    |\Phi^{\pm}\rangle  \coloneqq \frac{1}{\sqrt{2}}\left(|00\rangle \pm |11\rangle\right), \qquad |\Psi^{\pm}\rangle  \coloneqq \frac{1}{\sqrt{2}}\left(|01\rangle \pm |10\rangle\right) .
    \label{eq:bell-basis}
\end{equation}
This means that the SWAP observable for qubits can be measured by means of a Bell measurement and classical post-processing, a fact that is used in the destructive SWAP test method for measuring the SWAP observable \cite{GC13} (see also \cite{Brun04,Suba2019} and Eqs.~(26)--(37) of \cite{RASW23} for a review of this method). 

For convenience, we briefly review the destructive SWAP\ test \cite{GC13} for estimating the
overlap term $\operatorname{Tr}[\rho\sigma]$, where $\rho$ and $\sigma$ are
$n$-qubit states, with~$\rho$ a state of qubits $1$, \ldots, $n$ and $\sigma$
a state of qubits $n+1$, \ldots, $2n$. The idea behind it is a consequence of
the following observation:
\begin{align}
\operatorname{Tr}[\rho\sigma] &  =\operatorname{Tr}[\text{SWAP}^{(n)}
(\rho\otimes\sigma)] \label{eq:dest-SWAP-key-1} \\
&  =\sum_{\vec{k},\vec{\ell}\in\left\{  0,1\right\}  ^{n}}\left(  -1\right)
^{\vec{k}\cdot\vec{\ell}}\operatorname{Tr}[\Phi^{\vec{k}\vec{\ell}}\left(
\rho\otimes\sigma\right)  ],
\label{eq:dest-SWAP-key-2}
\end{align}
where
\begin{align}
\vec{k} &  \equiv(k_{1},k_{2},\ldots,k_{n}),\qquad\vec{\ell}\equiv(\ell
_{1},\ell_{2},\ldots,\ell_{n}),\\
\Phi^{\vec{k}\vec{\ell}} &  \equiv\Phi_{1,n+1}^{k_{1}\ell_{1}}\otimes
\Phi_{2,n+2}^{k_{2}\ell_{2}}\otimes\cdots\otimes\Phi_{n,2n}^{k_{n}\ell_{n}
},\label{eq:ordering-state-dest-SWAP-test}
\end{align}
and we used the identity in \eqref{eq:swap-bell}, as well as the fact that
\begin{align}
\text{SWAP}^{(n)} &  =\text{SWAP}^{\otimes n} \label{eq:Swap-tensor-prod}\\
&  =\left(  \sum_{k_{1},\ell_{1}}\left(  -1\right)  ^{k_{1}\ell_{1}}
\Phi_{1,n+1}^{k_{1}\ell_{1}}\right)  \otimes\cdots\otimes\left(  \sum
_{k_{n},\ell_{n}}\left(  -1\right)  ^{k_{n}\ell_{n}}\Phi_{n,2n}^{k_{n}\ell
_{n}}\right)  \\
&  =\sum_{\vec{k},\vec{\ell}\in\left\{  0,1\right\}  ^{n}}\left(  -1\right)
^{\vec{k}\cdot\vec{\ell}}\Phi^{\vec{k}\vec{\ell}}.
\label{eq:Swap-tensor-prod-3}
\end{align}

By setting $Z\equiv(\vec{K},\vec{L})$ to be a multi-indexed random variable
taking the value $\left(  -1\right)  ^{\vec{k}\cdot\vec{\ell}}$ with
probability
\begin{equation}
p(\vec{k},\vec{\ell})\coloneqq\operatorname{Tr}[\Phi^{\vec{k}\vec{\ell}
}\left(  \rho\otimes\sigma\right)  ],
\end{equation}
we find from \eqref{eq:dest-SWAP-key-1}--\eqref{eq:dest-SWAP-key-2} that its expectation is given by
\begin{equation}
\mathbb{E}[Z]=\sum_{\vec{k},\vec{\ell}\in\left\{  0,1\right\}  ^{n}}\left(
-1\right)  ^{\vec{k}\cdot\vec{\ell}}\operatorname{Tr}[\Phi^{\vec{k}\vec{\ell}
}\left(  \rho\otimes\sigma\right)  ]=\operatorname{Tr}[\rho\sigma].
\end{equation}
This observation then leads to the following quantum algorithm (destructive SWAP test) for estimating
$\operatorname{Tr}[\rho\sigma]$, within additive error $\varepsilon$ and with
success probability at least $1-\delta$, where $\varepsilon > 0$ and $\delta\in(0,1)$.

\begin{algorithm}
\label{alg:HS-state-estimate} Given are quantum circuits to prepare the $n$-qubit states
$\rho$ and$~\sigma$.

\begin{enumerate}
\item Fix $\varepsilon>0$ and $\delta\in(0,1)$. Set $T\geq\frac{2}
{\varepsilon^{2}}\ln\!\left(  \frac{2}{\delta}\right)  $ and set $t=1$.

\item Prepare the states $\rho$ and $\sigma$ on $2n$ qubits (using the
ordering specified in~\eqref{eq:ordering-state-dest-SWAP-test}).

\item Perform the Bell measurement $\{\Phi^{\vec{k}\vec{\ell}}\}_{\vec{k}
\vec{\ell}}$\ on the $2n$ qubits, which leads to the measurement outcomes $\vec{k}$
and $\vec{\ell}$.

\item Set $Z_{t}=\left(  -1\right)  ^{\vec{k}\cdot\vec{\ell}}$.

\item Increment $t$.

\item Repeat Steps 2.-5.~until $t>T$ and then output $\overline{Z}
\coloneqq\frac{1}{T}\sum_{t=1}^{T}Z_{t}$ as an estimate of $\operatorname{Tr}
[\rho\sigma]$.
\end{enumerate}
\end{algorithm}

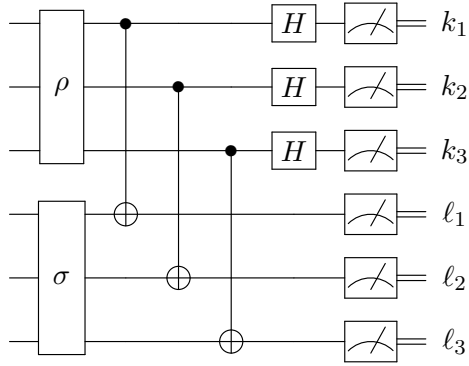
\begin{figure}
\centerline{\Qcircuit @C=1em @R=0.8em{
 & & \multigate{2}{\rho} & \ctrl{3} & \qw & \qw & \gate{H} & \meter & \cw & k_1 \\
  & &  \ghost{\rho} & \qw & \ctrl{3} & \qw & \gate{H} & \meter & \cw & k_2 \\
  & &  \ghost{\rho} & \qw & \qw & \ctrl{3} & \gate{H} & \meter & \cw & k_3 \\
 & & \multigate{2}{\sigma} & \targ & \qw & \qw & \qw & \meter & \cw & \ell_1 \\
 & &  \ghost{\sigma} & \qw & \targ & \qw & \qw & \meter & \cw & \ell_2 \\
 & & \ghost{\sigma} & \qw & \qw & \targ & \qw & \meter & \cw & \ell_3 
}}
\caption{\small Depiction of the core quantum subroutine given in Steps~2.-3.~of Algorithm~\ref{alg:HS-state-estimate}, for the three-qubit states $\rho$ and $\sigma$. This algorithm estimates the overlap~$\operatorname{Tr}[\rho\sigma]$.}
\label{fig:state-overlap-alg}
\end{figure}

Figure~\ref{fig:state-overlap-alg} depicts the core quantum subroutine of Algorithm~\ref{alg:HS-state-estimate}.
By the Hoeffding inequality (recalled as Theorem~\ref{thm:hoeffding} below), we are
guaranteed that the output of Algorithm~\ref{alg:HS-state-estimate} satisfies
\begin{equation}
\Pr\!\left[  \left\vert \overline{Z}-\operatorname{Tr}[\rho\sigma]\right\vert
\leq\varepsilon\right]  \geq1-\delta,
\end{equation}
due to the choice $T\geq\frac{2}{\varepsilon^{2}}\ln\!\left(  \frac{2}{\delta
}\right)  $.

Clearly, by the expansion in \eqref{eq:HS-expansion} and repeating Algorithm~\ref{alg:HS-state-estimate} three times, one can use $O\!\left(\frac{1}{\varepsilon^2} \ln\! \left(\frac{1}{\delta}\right)\right)$ samples of $\rho$ and $\sigma$ in order to obtain an estimate of~\eqref{eq:HS-expansion} within additive error $\varepsilon > 0$ and with success probability not smaller than $1-\delta$, where $\delta\in(0,1)$.

\begin{theorem}[Hoeffding Inequality \cite{H63}]
\label{thm:hoeffding}
Suppose that we are given $T$~independent samples $Y_1, \ldots, Y_T$ of a bounded random variable $Y$ taking values in the interval $[a,b]$ and having mean $\mu$. Set 
$
    \overline{Y_T} \coloneqq \frac{1}{T} (Y_1 + \cdots +Y_T)
$
to be the sample mean. Let $\varepsilon > 0$ be the desired accuracy, and let $1-\delta$ be the desired success probability, where $\delta \in (0,1)$. Then
\begin{equation}
\label{eq:Hoeff-bound}
\Pr[\vert \overline{Y_T} - \mu \vert \leq \varepsilon] \geq 1-\delta,
\end{equation}
as long as 
$
    T \geq \frac{M^2}{2\varepsilon^2} \ln \!\left( \frac{2}{\delta}\right),
$
where $M \coloneqq b -  a$.
\end{theorem}

\section{Quantum algorithms for testing symmetries}

\label{sec:symm-testing-algos}

\subsection{Testing symmetries of states}

\label{sec:symm-testing-states}

Let us now introduce a simple quantum algorithm for testing symmetry of
the state $\rho$ with respect to the unitary representation $\left\{
U(g)\right\}  _{g\in G}$ of a group~$G$. Specifically, the goal is to estimate
the normalized commutator norm in~\eqref{eq:state-sym-meas}. As discussed around \eqref{eq:state-sym-meas}, this
asymmetry measure is equal to zero if and only if $\left[  U(g),\rho\right]
=0$ for all $g\in G$. To start off, we establish the following lemma,
which provides a direct link between the asymmetry measure in \eqref{eq:state-sym-meas}, and an
approach we can use for estimating it on a quantum computer.

\begin{lemma}
    Given a state $\rho$ and a unitary representation $\{U(g)\}_{g\in G}$ of a group $G$, the following equality holds:
    \begin{equation}
        \frac{1}{\vert G \vert} \sum\limits_{g \in G} \left\Vert [U(g), \rho] \right\Vert^2_2 = 2 \left( \Tr[\rho^2] - \Tr[\rho  \mathcal{T}_G(\rho)] \right),
        \label{eq:HS-commutator-norm-expand}
    \end{equation}
    where $\mathcal{T}_G$ is the twirl channel given by 
\begin{equation}
    \mathcal{T}_G( \cdot ) \coloneqq  \frac{1}{\vert G \vert } \sum\limits_{g \in G} U(g) (\cdot) U(g)^\dagger.
\end{equation}
\end{lemma}

\begin{proof}
Consider the following equalities:
\begin{align}
    \label{eq:HS-Sym-Simpl}
    \left\Vert [U(g), \rho] \right\Vert^2_2 
    &= \left\Vert \rho U(g) - U(g)\rho \right\Vert^2_2  \\
    &= \left\Vert \rho - U(g) \rho U(g)^\dagger \right\Vert^2_2  \\
    &= \Tr[\rho^2] + \Tr[(U(g) \rho U(g)^\dagger)^2] - 2\Tr[\rho U(g) \rho U(g)^\dagger] \\
    &= 2\left( \Tr[\rho^2] - \Tr[\rho U(g) \rho U(g)^\dagger] \right),
\end{align}
where the second equality is due to the unitary invariance of the Hilbert--Schmidt norm, the third from the expansion in \eqref{eq:HS-expansion}, and the final one from cyclicity of trace. Thus, we see that
\begin{align}
    \frac{1}{\vert G \vert} \sum\limits_{g \in G} \left\Vert [U(g), \rho] \right\Vert^2_2 
    &= \frac{1}{\vert G \vert} \sum\limits_{g \in G} 2 \left( \Tr[\rho^2] - \Tr[\rho U(g) \rho U(g)^\dagger] \right)  \\
    &= 2 \left( \Tr[\rho^2] - \Tr[\rho \mathcal{T}_G(\rho)] \right)  ,
\end{align} 
concluding the proof.
\end{proof}

\medskip
Now suppose that the state $\rho$ is an $n$-qubit state and efficiently
preparable on a quantum computer, either by a quantum circuit or other means, and that,
for all $g\in G$,
there exists a circuit that efficiently realizes the $n$-qubit unitary~$U(g)$. Then the idea for estimating the asymmetry measure in~\eqref{eq:state-sym-meas} is simple: Perform the destructive SWAP test (Algorithm~\ref{alg:HS-state-estimate}) to estimate $\Tr[\rho^2]$ and perform the same test, using instead $\rho$ and its twirled version $\mathcal{T}_G(\rho)$, to estimate $\Tr[\rho \mathcal{T}_G(\rho)]$. When estimating the latter term, we modify Algorithm~\ref{alg:HS-state-estimate} to be as follows:

\begin{algorithm}
\label{alg:twirled-state-estimate} Given is a quantum circuit to prepare the $n$-qubit state
$\rho$ and circuits to generate the unitaries in the  representation $\{U(g)\}_{g\in G}$.

\begin{enumerate}
\item Fix $\varepsilon>0$ and $\delta\in(0,1)$. Set $T\geq\frac{2}
{\varepsilon^{2}}\ln\!\left(  \frac{2}{\delta}\right)  $ and set $t=1$.

\item Pick $g \in G$ uniformly at random. Prepare the states $\rho$ and $U(g) \rho U(g)^\dag$ on $2n$ qubits (using the
ordering specified in~\eqref{eq:ordering-state-dest-SWAP-test}).

\item Perform the Bell measurement $\{\Phi^{\vec{k}\vec{\ell}}\}_{\vec{k}
\vec{\ell}}$\ on the $2n$ qubits, which leads to the measurement outcomes $\vec{k}$
and $\vec{\ell}$.

\item Set $Z_{t}=\left(  -1\right)  ^{\vec{k}\cdot\vec{\ell}}$.

\item Increment $t$.

\item Repeat Steps 2.-5.~until $t>T$ and then output $\overline{Z}
\coloneqq\frac{1}{T}\sum_{t=1}^{T}Z_{t}$ as an estimate of $\Tr[\rho \mathcal{T}_G(\rho)]$.
\end{enumerate}
\end{algorithm}

Thus, by combining the estimates of $\Tr[\rho^2]$ and $\Tr[\rho \mathcal{T}_G(\rho)]$ according to~\eqref{eq:HS-commutator-norm-expand}, it follows that this approach uses $O\!\left(\frac{1}{\varepsilon^2} \ln\! \left(\frac{1}{\delta}\right)\right)$ samples of $\rho$ in order to obtain an estimate of the asymmetry measure in \eqref{eq:state-sym-meas} within additive error $\varepsilon >0$ and with success probability not smaller than $1-\delta$, where $\delta\in(0,1)$.

\subsection{Estimating the Hilbert--Schmidt distance of the Choi states of channels}

\label{sec:HS-Choi-states-algo}

Let us now introduce a method for
estimating the Hilbert--Schmidt distance between the Choi states of two quantum channels, as a generalization of the destructive SWAP test used for estimating the Hilbert--Schmidt distance between two states. This algorithm
has applications beyond symmetry testing, for example, in quantum channel compilation
as a generalization of compiling states (see~\cite{Ezzell_2023} for the latter).

To begin with, recall that two channels $\mathcal{N}$ and
$\mathcal{M}$ are equal if and only if their Choi states are equal \cite[Section~4.4.2]{wilde_2017}; i.e.,
\begin{equation}
\mathcal{N}=\mathcal{M}\qquad\Leftrightarrow\qquad\Phi^{\mathcal{N}}
=\Phi^{\mathcal{M}},
\end{equation}
where the Choi states $\Phi^{\mathcal{N}}$ and $
\Phi^{\mathcal{M}}$ are defined in \eqref{eq:choi-state-def}.
One way to determine whether the equality above holds approximately is to employ
the Hilbert--Schmidt distance of the Choi states:
\begin{equation}
\left\Vert \Phi^{\mathcal{N}}-\Phi^{\mathcal{M}}\right\Vert _{2},
\end{equation}
where the Hilbert--Schmidt norm is defined in \eqref{eq:HS-norm}.
This is due to the positive definiteness or faithfulness of the norm, i.e.,
\begin{equation}
\left\Vert \Phi^{\mathcal{N}}-\Phi^{\mathcal{M}}\right\Vert
_{2}=0\qquad\Leftrightarrow\qquad\Phi^{\mathcal{N}}=\Phi^{\mathcal{M}}.
\end{equation}
Using the expansion in \eqref{eq:HS-expansion}, consider that
\begin{equation}
\left\Vert \Phi^{\mathcal{N}}-\Phi^{\mathcal{M}}\right\Vert _{2}
^{2}=\operatorname{Tr}[(\Phi^{\mathcal{N}})^{2}]+\operatorname{Tr}
[(\Phi^{\mathcal{M}})^{2}]-2\operatorname{Tr}[\Phi^{\mathcal{N}}
\Phi^{\mathcal{M}}]. \label{eq:HS-expand}
\end{equation}
The following lemma gives a way of rewriting the overlap $
\operatorname{Tr}[\Phi^{\mathcal{N}}\Phi^{\mathcal{M}}]$ in terms of the SWAP observable, and it is critical to our simplified approach for estimating the Hilbert--Schmidt distance between the Choi states of two channels.

\begin{lemma}
\label{lem:key-identity-SWAP}
Let $\mathcal{N}$ and $\mathcal{M}$ be channels
with Choi states $\Phi^{\mathcal{N}}$ and $\Phi^{\mathcal{M}}$,
respectively, and $d$-dimensional inputs. Then
\begin{equation}
\operatorname{Tr}[\Phi^{\mathcal{N}}\Phi^{\mathcal{M}}]=\frac
{1}{d^{2}}\operatorname{Tr}[\operatorname{SWAP}(\mathcal{N\otimes
M})(\operatorname{SWAP})].
\end{equation}

\end{lemma}

\begin{proof}
Consider that
\begin{align}
\operatorname{Tr}[\Phi^{\mathcal{N}}\Phi^{\mathcal{M}}]  &
=\operatorname{Tr}[(\operatorname{id}\otimes\mathcal{N})(\Phi^d)(\operatorname{id}\otimes\mathcal{M})(\Phi^d)] \label{eq:Choi-overlap-pf-1} \\
&  =\frac{1}{d^{2}}\sum_{i,j,k,\ell}\operatorname{Tr}[\left(  |i\rangle\!\langle
j|\otimes\mathcal{N}(|i\rangle\!\langle j|)\right)  \left(  |k\rangle
\!\langle\ell|\otimes\mathcal{M}(|k\rangle\!\langle\ell|)\right)  ]\\
&  =\frac{1}{d^{2}}\sum_{i,j,k,\ell}\langle\ell|i\rangle\!\langle j|k\rangle
\otimes\operatorname{Tr}[\mathcal{N}(|i\rangle\!\langle j|)\mathcal{M}
(|k\rangle\!\langle\ell|)]\\
&  =\frac{1}{d^{2}}\sum_{i,j}\operatorname{Tr}[\mathcal{N}(|i\rangle\!\langle
j|)\mathcal{M}(|j\rangle\!\langle i|)]
\label{eq:Choi-overlap-pf-next} 
\\
&  =\frac{1}{d^{2}}\sum_{i,j}\operatorname{Tr}[\operatorname{SWAP}\left(
\mathcal{N}\otimes\mathcal{M}\right)  (|i\rangle\!\langle j|\otimes
|j\rangle\!\langle i|)]\\
&  =\frac{1}{d^{2}}\operatorname{Tr}[\operatorname{SWAP}\left(  \mathcal{N}
\otimes\mathcal{M}\right)  \left(  \operatorname{SWAP}\right)  ].
\end{align}
The penultimate equality follows from \eqref{eq:swap-prod-id}.
\end{proof}

\medskip

Now suppose that the channels $\mathcal{N}$ and $\mathcal{M}$ each accept $n$
qubits as input and output $m$ qubits. Then each of the terms in
\eqref{eq:HS-expand}\ can be efficiently measured on a quantum computer. For
example, to measure the last term $\operatorname{Tr}[\Phi^{\mathcal{N}}\Phi^{\mathcal{M}}]$, one could prepare the tensor-product state $\Phi^{\mathcal{N}}\otimes\Phi^{\mathcal{M}}$ and then perform a
destructive SWAP test, as recalled in Algorithm~\ref{alg:HS-state-estimate}. This approach, which we consider to be a naive approach in light of Algorithm~\ref{alg:HS-estimate} below, requires $2(n+m)$ qubits in total, for a circuit width of $2(n+m)$ qubits.
However, what follows as a consequence of Lemma~\ref{lem:key-identity-SWAP} is that there is a simpler procedure for
estimating $\operatorname{Tr}[\Phi^{\mathcal{N}}\Phi^{\mathcal{M}}]$, which requires preparing only $2n$ qubits at the input and acting on
$2m$ qubits at the output, and thus for a circuit width of $\max\{2n,2m\}$ qubits. 

Indeed, Lemma~\ref{lem:key-identity-SWAP} establishes that
\begin{equation}
\operatorname{Tr}[\Phi^{\mathcal{N}}\Phi^{\mathcal{M}}]=\frac
{1}{2^{2n}}\operatorname{Tr}[\operatorname{SWAP}^{(m)}\left(  \mathcal{N}
\otimes\mathcal{M}\right)  (\operatorname{SWAP}^{(n)})],
\label{eq:HS-overlap-qubits}
\end{equation}
where the superscript notation explicitly indicates the number of qubits on
which the swap operator acts. Next recall \eqref{eq:Swap-tensor-prod}--\eqref{eq:Swap-tensor-prod-3}, 
which implies that
\begin{multline}
\frac{1}{2^{2n}}\operatorname{Tr}[\operatorname{SWAP}^{(m)}\left(
\mathcal{N}\otimes\mathcal{M}\right)  (\operatorname{SWAP}^{(n)})]\\
=\frac{1}{2^{2n}}\sum_{\vec{i},\vec{j}\in\left\{  0,1\right\}  ^{m}}\sum
_{\vec{k},\vec{\ell}\in\left\{  0,1\right\}  ^{n}}\left(  -1\right)  ^{\vec
{i}\cdot\vec{j}+\vec{k}\cdot\vec{\ell}}\operatorname{Tr}[\Phi^{\vec{i}\vec{j}
}\left(  \mathcal{N}\otimes\mathcal{M}\right)  (\Phi^{\vec{k}\vec{\ell}})],
\label{eq:SWAP-to-bell-reduction}
\end{multline}
where
\begin{align}
\vec{i}  &  \equiv(i_{1},i_{2},\ldots,i_{m}),  \qquad
\vec{j}    \equiv(j_{1},j_{2},\ldots,j_{m}),\\
\vec{k}  &  \equiv(k_{1},k_{2},\ldots,k_{n}), \qquad 
\vec{\ell}  \equiv(\ell_{1},\ell_{2},\ldots,\ell_{n}),\\
\Phi^{\vec{i}\vec{j}}  &  \equiv\Phi_{1,m+1}^{i_{1}j_{1}}\otimes\Phi
_{2,m+2}^{i_{2}j_{2}}\otimes\cdots\otimes\Phi_{m,2m}^{i_{m}j_{m}},\\
\Phi^{\vec{k}\vec{\ell}}  &  \equiv\Phi_{1,n+1}^{k_{1}\ell_{1}}\otimes
\Phi_{2,n+2}^{k_{2}\ell_{2}}\otimes\cdots\otimes\Phi_{n,2n}^{k_{n}\ell_{n}} .
\label{eq:input-bell-state}
\end{align}

Eq.~\eqref{eq:SWAP-to-bell-reduction} and Lemma~\ref{lem:key-identity-SWAP} are the key insights that lead to a simplified quantum algorithm for estimating the term $ \operatorname{Tr}[\Phi^{\mathcal{N}}\Phi^{\mathcal{M}}]$, which requires only $2n$ qubits at the input and $2m$ qubits at the output.
In the above, we have implicitly used the following ordering: the channel
$\mathcal{N}$ acts on input qubits $1,\ldots,n$ and produces output qubits
$1,\ldots,m$, the channel $\mathcal{M}$ acts on input qubits $n+1,\ldots,2n$
and produces output qubits $m+1,\ldots,2m$, and the qubits for the Bell states
are labeled as subscripts above. By setting $Y\equiv(\vec{I},\vec{J},\vec
{K},\vec{L})$ to be a multi-indexed random variable taking the value~$\left(
-1\right)  ^{\vec{i}\cdot\vec{j}+\vec{k}\cdot\vec{\ell}}$ with probability
\begin{equation}
p(\vec{k},\vec{\ell},\vec{i},\vec{j})    =p(\vec{i},\vec{j}|\vec{k},\vec
{\ell})\, p(\vec{k},\vec{\ell}),
\end{equation}
where
\begin{align}
p(\vec{k},\vec{\ell})    & \coloneqq \frac{1}{2^{2n}},\\
p(\vec{i},\vec{j}|\vec{k},\vec{\ell})  &  \coloneqq \operatorname{Tr}
[\Phi^{\vec{i}\vec{j}}\left(  \mathcal{N}\otimes\mathcal{M}\right)
(\Phi^{\vec{k}\vec{\ell}})],
\end{align}
we find from \eqref{eq:HS-overlap-qubits}--\eqref{eq:SWAP-to-bell-reduction} that its expectation is given by
\begin{equation}
\mathbb{E}[Y]=\frac{1}{2^{2n}}\operatorname{Tr}[\operatorname{SWAP}
^{(m)}\left(  \mathcal{N}\otimes\mathcal{M}\right)  (\operatorname{SWAP}
^{(n)})] = \operatorname{Tr}[\Phi^{\mathcal{N}}\Phi^{\mathcal{M}}].
\label{eq:expect-Y-Choi-overlap}
\end{equation}

The observation in \eqref{eq:expect-Y-Choi-overlap} then leads to the following quantum algorithm for estimating
$\operatorname{Tr}[\Phi^{\mathcal{N}}\Phi^{\mathcal{M}}]$, within additive error $\varepsilon$ and with success probability not smaller than $1-\delta$, where $\varepsilon>0$ and $\delta \in (0,1)$.

\begin{algorithm}
\label{alg:HS-estimate} Given are quantum circuits to implement the channels
$\mathcal{N}$ and~$\mathcal{M}$.

\begin{enumerate}
\item Fix $\varepsilon>0$ and $\delta\in(0,1)$. Set $T\geq\frac{2}{\varepsilon^{2}}\ln\!\left(  \frac{2}{\delta}\right)  $ and set $t=1$.

\item Generate the bit vectors $\vec{k}$ and $\vec{\ell}$ uniformly at random.

\item Prepare the Bell state $\Phi^{\vec{k}\vec{\ell}}$ on $2n$ qubits (using
the ordering specified in~\eqref{eq:input-bell-state}).

\item Apply the tensor-product channel $\mathcal{N}\otimes\mathcal{M}$ (using
the ordering specified after \eqref{eq:input-bell-state}).

\item Perform the Bell measurement $\{\Phi^{\vec{i}\vec{j}}\}_{\vec{i},\vec
{j}}$\ on the $2m$ output qubits, which leads to the measurement outcomes $\vec{i}$
and $\vec{j}$.

\item Set $Y_{t}=\left(  -1\right)  ^{\vec{i}\cdot\vec{j}+\vec{k}\cdot
\vec{\ell}}$.

\item Increment $t$.

\item Repeat Steps 2.-7.~until $t>T$ and then output $\overline{Y}
\coloneqq \frac{1}{T}\sum_{t=1}^{T}Y_{t}$ as an estimate of $\operatorname{Tr}
[\Phi^{\mathcal{N}}\Phi^{\mathcal{M}}]$.
\end{enumerate}
\end{algorithm}

Figure~\ref{fig:Choi-overlap-alg} depicts the core quantum subroutine of Algorithm~\ref{alg:HS-estimate}. By the Hoeffding inequality (recalled as Theorem~\ref{thm:hoeffding}), we are guaranteed that the output of
Algorithm~\ref{alg:HS-estimate} satisfies
\begin{equation}
\Pr\!\left[\left\vert \overline{Y}-\operatorname{Tr}[\Phi^{\mathcal{N}}\Phi^{\mathcal{M}}]\right\vert \leq\varepsilon\right]\geq 1-\delta,
\end{equation}
due to the choice $T\geq\frac{2}{\varepsilon^{2}}\ln\!\left(  \frac{2}{\delta}\right)  $.

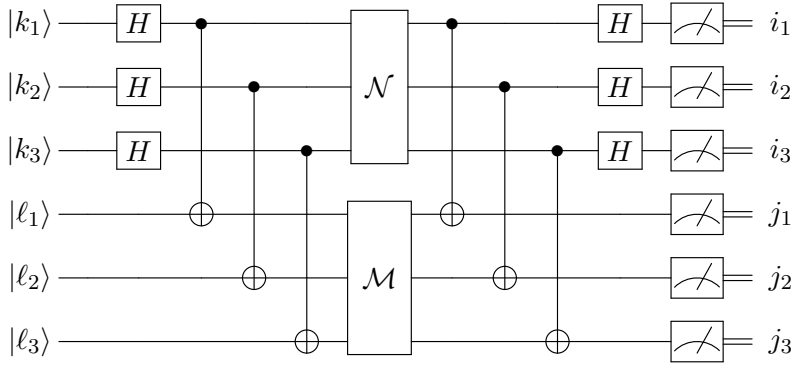
\begin{figure}
\centerline{
\Qcircuit @C=1em @R=0.8em{
|k_1\rangle & & \qw & \gate{H} & \ctrl{3} & \qw & \qw & \multigate{2}{\mathcal{N}} & \ctrl{3} & \qw & \qw & \gate{H}  & \meter & \cw & i_1 \\
|k_2\rangle  & & \qw & \gate{H} & \qw & \ctrl{3} & \qw &  \ghost{\mathcal{N}} & \qw & \ctrl{3} & \qw & \gate{H} & \meter & \cw & i_2 \\
|k_3\rangle  & & \qw & \gate{H} & \qw & \qw & \ctrl{3} &  \ghost{\mathcal{N}} & \qw & \qw & \ctrl{3} & \gate{H} & \meter & \cw & i_3 \\
|\ell_1 \rangle & & \qw  & \qw & \targ & \qw  & \qw & \multigate{2}{\mathcal{M}} & \targ & \qw & \qw & \qw & \meter & \cw & j_1 \\
|\ell_2 \rangle & & \qw & \qw & \qw & \targ  & \qw & \ghost{\mathcal{M}} & \qw & \targ & \qw & \qw & \meter & \cw & j_2 \\
|\ell_3 \rangle & & \qw & \qw & \qw & \qw & \targ  & \ghost{\mathcal{M}} & \qw & \qw & \targ & \qw & \meter & \cw & j_3 
}
}
\caption{\small Depiction of the core quantum subroutine given in Steps~2.-5.~of Algorithm~\ref{alg:HS-estimate}, such that the quantum channels $\mathcal{N}$ and $\mathcal{M}$ have three-qubit inputs and outputs. This algorithm estimates the overlap $\operatorname{Tr}[\Phi^{\mathcal{N}}\Phi^{\mathcal{M}}]$ of the Choi states of the channels. In this example, the algorithm begins by preparing the classical state $|k_1, k_2, k_3, \ell_1, \ell_2, \ell_3\rangle$, where the values $k_1, k_2, k_3, \ell_1, \ell_2, \ell_3$ are chosen uniformly at random, followed by a sequence of controlled NOTs and Hadamards. Before the channels are applied, the state is thus $|\Phi^{\vec{k} \vec{\ell}}\rangle$, as described in Algorithm~\ref{alg:HS-estimate}. After the channels are applied, Bell measurements are performed, which lead to the classical bit string $i_1 i_2 i_3 j_1 j_2 j_3$. In the diagram, we depict the realization of the channels $\mathcal{N}$ and $\mathcal{M}$ as black boxes, but in a simulation of them, one might make use of additional environment qubits that are prepared and then discarded.}
\label{fig:Choi-overlap-alg}
\end{figure}

By employing Algorithm~\ref{alg:HS-estimate} three times, we can thus estimate
\eqref{eq:HS-expand} within additive error $\varepsilon$ and with success probability not smaller than $1-\delta$, by using $O\!\left(\frac{1}{\varepsilon^2} \ln\! \left(\frac{1}{\delta}\right)\right)$ samples of the channels $\mathcal{N}$ and $\mathcal{M}$.

\subsection{Testing symmetries of channels}

\label{sec:symm-tests-channels-algos}

In this section, we leverage the methods for estimating the Hilbert--Schmidt asymmetry measure for states (Section~\ref{sec:symm-testing-states}), as well as the method for estimating the Hilbert--Schmidt distance between the Choi states of channels (Section~\ref{sec:HS-Choi-states-algo}), in order to develop an approach for estimating the covariance symmetry of a quantum channel $\mathcal{N}$ with respect to a unitary channel representation~$\{\mathcal{U}(g)\}_{g \in G}$.

Recalling the superoperator commutator notation defined in \eqref{eq:superop-comm-def}, we are interested in estimating the following asymmetry measure:
\begin{equation}
\frac{1}{\left\vert G\right\vert
}\sum_{g\in G}\left\Vert \left(  \operatorname{id}\otimes\left[
\mathcal{U}(g),\mathcal{N}\right]  \right)  (\Phi^d)\right\Vert
_{2}^{2}.
\label{eq:ch-sym-meas-def}
\end{equation}
As discussed around \eqref{eq:channel-sym-meas}, this asymmetry measure is equal to zero if and only if $\mathcal{N} \circ \mathcal{U}(g) = \mathcal{U}(g) \circ \mathcal{N}$ holds for every $g \in G$.

We begin with the following lemma:
\begin{lemma}
\label{lem:ch-sym-meas-reduction}
    Given a quantum channel $\mathcal{N}$ and a unitary channel representation $\{\mathcal{U}(g)\}_{g \in G}$, the following equality holds:
    \begin{multline}
        \frac{1}{\left\vert G\right\vert
}\sum_{g\in G}\left\Vert \left(  \operatorname{id}\otimes\left[
\mathcal{U}(g),\mathcal{N}\right]  \right)  (\Phi^d)\right\Vert
_{2}^{2} = \frac{2}{d^{2}}\operatorname{Tr}
[\operatorname{SWAP}(\mathcal{N}\mathcal{\otimes}\mathcal{N})(\operatorname{SWAP})]\\
-\frac{2}{d^{2}}\operatorname{Tr}\!\left[  \operatorname{SWAP}\left(  \frac
{1}{\left\vert G\right\vert }\sum_{g\in G}\left(  \mathcal{U}(g)\circ
\mathcal{N}\right)  \mathcal{\otimes}\left(  \mathcal{N}
\circ\mathcal{U}(g)\right)  \right)  (\operatorname{SWAP})\right].
\label{eq:ch-sym-meas-swaps}
    \end{multline}
\end{lemma}

\begin{proof}
Consider that, for all $g \in G$,
\begin{align}
& \left\Vert \left(  \operatorname{id}\otimes\left[
\mathcal{U}(g),\mathcal{N}\right]  \right)  (\Phi^d)\right\Vert
_{2}^{2} \notag \\
&  = \left\Vert \Phi^{\mathcal{U}(g)\circ \mathcal{N}}-\Phi^{\mathcal{N}\circ\mathcal{U}(g)}\right\Vert _{2}^{2} \\
&  =\operatorname{Tr}[(\Phi^{\mathcal{U}(g)\circ \mathcal{N}}
)^{2}]+\operatorname{Tr}[(\Phi^{\mathcal{N}\circ\mathcal{U}(g)}
)^{2}]-2\operatorname{Tr}[\Phi^{\mathcal{U}(g)\circ \mathcal{N}}
\Phi^{\mathcal{N}\circ\mathcal{U}(g)}]\\
&  =2\left(  \operatorname{Tr}[(\Phi^{\mathcal{N}})^{2}
]-\operatorname{Tr}[\Phi^{\mathcal{U}(g)\circ \mathcal{N}}\Phi^{\mathcal{N}\circ\mathcal{U}(g)}]\right)  ,
\end{align}
where we made use of the expansion in \eqref{eq:HS-expansion}, as well as the equalities
\begin{equation}
\operatorname{Tr}[(\Phi^{\mathcal{U}(g)\circ \mathcal{N}})^{2}] 
=\operatorname{Tr}[(\Phi^{\mathcal{N}})^{2}],\qquad 
\operatorname{Tr}[(\Phi^{\mathcal{N}\circ\mathcal{U}(g)})^{2}] 
=\operatorname{Tr}[(\Phi^{\mathcal{N}})^{2}].
\label{eq:simplifying-eqs-ch-meas}
\end{equation}
The equalities in \eqref{eq:simplifying-eqs-ch-meas} follow because
\begin{align}
\operatorname{Tr}[(\Phi^{\mathcal{U}(g)\circ\mathcal{N}})^{2}]  &
=\operatorname{Tr}[\{  (\operatorname{id}\otimes(\mathcal{U}
(g)\circ\mathcal{N}))(\Phi^{d})\}  ^{2}]\\
& =\operatorname{Tr}[\{  (\operatorname{id}\otimes\mathcal{N})(\Phi
^{d})\}  ^{2}]\label{eq:unitary-invariance-1}\\
& =\operatorname{Tr}[(\Phi^{\mathcal{N}})^{2}],\\
\operatorname{Tr}[(\Phi^{\mathcal{N}\circ\mathcal{U}(g)})^{2}]  &
=\operatorname{Tr}[\{  (\operatorname{id}\otimes(\mathcal{N}
\circ\mathcal{U}(g)))(\Phi^{d})\}  ^{2}]\\
& =\operatorname{Tr}[\{  (\mathcal{U}^{T}(g)\otimes\mathcal{N})(\Phi
^{d})\}  ^{2}] \label{eq:transpose-trick-1} \\
& =\operatorname{Tr}[\{  (\operatorname{id}\otimes\mathcal{N})(\Phi
^{d})\}  ^{2}]\label{eq:unitary-invariance-2}\\
& =\operatorname{Tr}[(\Phi^{\mathcal{N}})^{2}].
\end{align}
The equalities in \eqref{eq:unitary-invariance-1} and
\eqref{eq:unitary-invariance-2} in turn follow because the function $\operatorname{Tr}
[\sigma^{2}]$ depends only on the eigenvalues of $\sigma$, and its eigenvalues
are invariant under the action of a unitary channel. The equality in \eqref{eq:transpose-trick-1} follows
from the transpose trick \cite[Exercise~3.7.12]{wilde_2017}; i.e., the identity $(\operatorname{id}
\otimes\mathcal{U})(\Phi^{d})=(\mathcal{U}^{T}\otimes\operatorname{id}
)(\Phi^{d})$ holds for every unitary channel $\mathcal{U}$, where the
transpose channel is defined as $\mathcal{U}^{T}(\cdot)=U^{T}(\cdot
)\overline{U}$, with $\overline{U}$ the matrix realized from $U$ by entrywise
complex conjugation.
Now employing Lemma~\ref{lem:key-identity-SWAP}, we can write
\begin{align}
\operatorname{Tr}[(\Phi^{\mathcal{N}})^{2}] &  =\frac{1}{d^{2}
}\operatorname{Tr}[\operatorname{SWAP}(\mathcal{N}\mathcal{\otimes
}\mathcal{N})(\operatorname{SWAP})],
\\
\operatorname{Tr}[\Phi^{\mathcal{U}(g)\circ \mathcal{N}}\Phi
^{\mathcal{N}\circ\mathcal{U}(g)}] &  =\frac{1}{d^{2}}
\operatorname{Tr}[\operatorname{SWAP}(\left(  \mathcal{U}(g)\circ
\mathcal{N}\right)  \mathcal{\otimes}\left(  \mathcal{N}
\circ\mathcal{U}(g)\right)  )(\operatorname{SWAP})],
\end{align}
which finally implies the claim in \eqref{eq:ch-sym-meas-swaps}.
\end{proof}

\medskip
In order to estimate the channel asymmetry measure in \eqref{eq:ch-sym-meas-def}, it follows from Lemma~\ref{lem:ch-sym-meas-reduction} that 
we can make use of Algorithm~\ref{alg:HS-estimate} to estimate the following two
quantities:
\begin{align}
& \frac{1}{d^{2}}\operatorname{Tr}
[\operatorname{SWAP}(\mathcal{N}\mathcal{\otimes}\mathcal{N})(\operatorname{SWAP})],
\label{eq:1st-quant-meas-ch-sym}\\
& \frac{1}{d^{2}}\operatorname{Tr}\!\left[  \operatorname{SWAP}\left(  \frac
{1}{\left\vert G\right\vert }\sum_{g\in G}\left(  \mathcal{U}(g)\circ
\mathcal{N}\right)  \mathcal{\otimes}\left(  \mathcal{N}
\circ\mathcal{U}(g)\right)  \right)  (\operatorname{SWAP})\right],
\label{eq:2nd-quant-meas-ch-sym}
\end{align}
subtract the estimates, and multiply by two.
For estimating the quantity in~\eqref{eq:2nd-quant-meas-ch-sym}, similar to how we did in Algorithm~\ref{alg:twirled-state-estimate}, we can slightly revise
Algorithm~\ref{alg:HS-estimate} such that $g\in G$ is chosen uniformly at random in each step.

\begin{remark}
\label{rem:extension-ch-symm}
More generally, a quantum channel $\mathcal{N}$ can possess a covariance symmetry of the following form:
\begin{equation}
    \mathcal{N} \circ \mathcal{U}(g) = \mathcal{V}(g) \circ \mathcal{N} \qquad \forall g \in G,
\end{equation}
where $\{\mathcal{U}(g)\}_{g\in G}$ and $\{\mathcal{V}(g)\}_{g\in G}$ are unitary channel representations of a group $G$. This more general symmetry occurs especially in the case in which the dimensions of the channel input and output differ (as is the case, e.g., for the quantum erasure channel \cite{wilde_2017}).

We note here that all of the observations from this section apply to this more general case. Namely, the asymmetry measure from \eqref{eq:ch-sym-meas-def} generalizes to
\begin{multline}
        \frac{1}{\left\vert G\right\vert
}\sum_{g\in G}\left\Vert \Phi^{\mathcal{N} \circ \mathcal{U}(g)} - \Phi^{\mathcal{V}(g) \circ \mathcal{N}} \right\Vert
_{2}^{2} 
= \frac{2}{d^{2}}\operatorname{Tr}
[\operatorname{SWAP}(\mathcal{N}\mathcal{\otimes}\mathcal{N})(\operatorname{SWAP})]\\
-\frac{2}{d^{2}}\operatorname{Tr}\!\left[  \operatorname{SWAP}\left(  \frac
{1}{\left\vert G\right\vert }\sum_{g\in G}\left(  \mathcal{V}(g)\circ
\mathcal{N}\right)  \mathcal{\otimes}\left(  \mathcal{N}
\circ\mathcal{U}(g)\right)  \right)  (\operatorname{SWAP})\right],
    \end{multline}
    where the equality follows from essentially the same proof given for Lemma~\ref{lem:ch-sym-meas-reduction}. Then we can again make use of Algorithm~\ref{alg:HS-estimate}, in a similar fashion as discussed around \eqref{eq:2nd-quant-meas-ch-sym}, in order to estimate the asymmetry measure above.
\end{remark}

\subsection{Testing symmetries of Lindbladians}

\label{sec:symm-tests-lindbladians-algos}

In this section, we apply the symmetry testing algorithm from Section~\ref{sec:symm-tests-channels-algos} to the task of measuring the symmetry of a Lindbladian $\mathcal{L}$, as defined in \eqref{eq:lindblad-master-equation}. Given that the channel realized by the master equation in \eqref{eq:lindblad-master-equation} is $e^{\mathcal{L}t}$, our basic idea is to test for symmetry of this channel by means of the algorithm from Section~\ref{sec:symm-tests-channels-algos}. As discussed previously, this amounts to estimating the two terms in \eqref{eq:1st-quant-meas-ch-sym} and \eqref{eq:2nd-quant-meas-ch-sym} using Algorithm~\ref{alg:HS-estimate}, but with the replacement $\mathcal{N} \to e^{\mathcal{L}t}$, and combining the estimates according to \eqref{eq:ch-sym-meas-swaps}. The result is to form an estimate of  the following asymmetry measure:
\begin{equation}
a(\mathcal{L},t,\{U(g)\}_{g\in G})\coloneqq \frac{1}{\left\vert G\right\vert
}\sum_{g\in G}\left\Vert \left(  \operatorname{id}\otimes\left[
\mathcal{U}(g),e^{\mathcal{L}t}\right]  \right)  (\Phi^d)\right\Vert
_{2}^{2},
\label{eq:asymmetry-measure}
\end{equation}
In order to do so, we require a means by which the channel $e^{\mathcal{L}t}$ can be realized or simulated. We can accomplish the latter by employing one of several quantum algorithms for simulating Lindbladian evolutions \cite{Childs2016EfficientDynamics,Cleve2016EfficientEvolution,KSMM22,Schlimgen2022QuantumOperators,Suri2022Two-UnitarySimulation} (see \cite{Miessen2022QuantumDynamics} for a review).

The basic condition for symmetry of a Lindbladian $\mathcal{L}$ with respect to a unitary channel representation is as follows \cite{Holevo1993,Holevo1996,holevo1996covariant}:
\begin{equation}
\mathcal{L}\circ\mathcal{U}(g)=\mathcal{U}(g)\circ\mathcal{L}\qquad\forall
g\in G. \label{eq:lindbladian-symmetry}
\end{equation}
An alternative definition for symmetry of a Lindbladian $\mathcal{L}$ with respect to a unitary channel representation $\{\mathcal{U}(g)\}_{g\in G}$ is similar to what we defined in \eqref{eq:channel-symm-def-1}--\eqref{eq:channel-symm-def-2}, for channel symmetry \cite{Holevo1993,Holevo1996,holevo1996covariant}:
\begin{equation}
e^{\mathcal{L}t}\circ\mathcal{U}(g)=\mathcal{U}(g)\circ e^{\mathcal{L}t} \qquad  \forall t\in\mathbb{R}, g\in G
\label{eq:def-Lindblad-sym} .
\end{equation}
In the following proposition, we recall the well known fact that these two definitions are actually equivalent:

\begin{proposition}
The symmetry condition in \eqref{eq:def-Lindblad-sym} holds if and only if it
holds for the Lindbladian $\mathcal{L}$, so that

\end{proposition}

\begin{proof}
Suppose that \eqref{eq:def-Lindblad-sym} holds. We then find that
\begin{equation}
\left.  \frac{\partial}{\partial t}\left(  e^{\mathcal{L}t}\circ
\mathcal{U}(g)\right)  \right\vert _{t=0}=\left.  \frac{\partial}{\partial
t}\left(  \mathcal{U}(g)\circ e^{\mathcal{L}t}\right)  \right\vert _{t=0}.
\end{equation}
The left-hand side then evaluates to $\mathcal{L}\circ\mathcal{U}(g)$ and the
right-hand side to $\mathcal{U}(g)\circ\mathcal{L}$, concluding the proof of
the if-part of the proposition. To see the other implication (the only-if
part), suppose that \eqref{eq:lindbladian-symmetry} holds. Then
\begin{align}
e^{\mathcal{L}t}\circ\mathcal{U}(g)    =\sum_{\ell=0}^{\infty}\frac{\left(
\mathcal{L}^{\ell}\circ\mathcal{U}(g)\right)  t^{\ell}}{\ell!}
 =\sum_{\ell=0}^{\infty}\frac{\left(  \mathcal{U}(g)\circ\mathcal{L}^{\ell
}\right)  t^{\ell}}{\ell!}
  =\mathcal{U}(g)\circ e^{\mathcal{L}t},
\end{align}
where the second equality follows from repeated application of \eqref{eq:lindbladian-symmetry}.
\end{proof}

\medskip
In fact, the main finding of \cite{Holevo1993} establishes a much stronger result: the symmetry condition in \eqref{eq:lindbladian-symmetry} is equivalent to the existence of a representation of $\mathcal{L}$ of the form in \eqref{eq:lindblad-master-equation}, such that the completely positive map $(\cdot) \to \sum_k L_k (\cdot) L_k^\dag $ is covariant with respect to $\{\mathcal{U}(g)\}_{g\in G}$ and $[U(g), H] = 0 $ for all $g \in G$.

For small $t$, we perform a Taylor expansion of the Lindbladian term contained in the asymmetry measure as defined in \eqref{eq:asymmetry-measure}, in order to elucidate a relation between approximate symmetry of the channel $e^{\mathcal{L}t}$ and the Lindbladian $\mathcal{L}$:
\begin{align}
    &\frac{1}{|G|}\sum_{g \in G}\left\Vert \left( \operatorname{id} \otimes \left[\mathcal{U}(g), e^{\mathcal{L}t}\right] \right) (\Phi^{d}) \right\Vert_{2}^{2} \\ 
    &= \frac{1}{|G|}\sum_{g \in G}\left\Vert \left( \operatorname{id} \otimes [\mathcal{U}(g), \operatorname{id} + \mathcal{L}t + O(t^2)] \right) (\Phi^d) \right\Vert_{2}^{2} \\
    &= \frac{1}{|G|}\sum_{g \in G}\left\Vert \left( \operatorname{id} \otimes \left([\mathcal{U}(g), \operatorname{id}] + [\mathcal{U}(g), \mathcal{L}t] + [\mathcal{U}(g), O(t^2)] \right)\right) (\Phi^d) \right\Vert_{2}^{2} \\
    &= \frac{1}{|G|}\sum_{g \in G}\left\Vert \left( \operatorname{id} \otimes \left([\mathcal{U}(g), \mathcal{L}]t + [\mathcal{U}(g), O(t^2)] \right)\right) (\Phi^d) \right\Vert_{2}^{2} \\
    &= \frac{t^2}{|G|}\sum_{g \in G}\left\Vert \left( \operatorname{id} \otimes [\mathcal{U}(g), \mathcal{L}] \right) (\Phi^d) \right\Vert_{2}^{2} + O(t^3).
\end{align}

\section{Simulations}

\label{sec:sims}

In this section, we first describe two examples of open quantum systems, namely, the amplitude damping channel and a two-qubit spin chain. We subsequently present simulation results obtained from Qiskit implementations of the aforementioned systems, wherein we test them for symmetry with respect to the finite discrete group $\mathbb{Z}_2$.\footnote{All code used to run simulations, generate plots, and perform proof-related calculations is available at \url{https://github.com/radulaski/SymmetryTestingQuantumAlgorithms}.}

In the case of the amplitude damping channel, we use the algorithm discussed around \eqref{eq:1st-quant-meas-ch-sym}--\eqref{eq:2nd-quant-meas-ch-sym}  to estimate the asymmetry measure in \eqref{eq:ch-sym-meas-swaps} and then plot the same as a function of $\Gamma t$, where $\Gamma$ represents the rate of decay per unit time and $t$ denotes time. We find that, for all values of $\Gamma t$, when testing for $Z$ symmetry (i.e., when our chosen unitary group representation for $\mathbb{Z}_2$ is $\{ U(g) \}_{g \in \mathbb{Z}_2} = \{ I, Z \}$), the asymmetry measure is approximately equal to zero with accuracy $\epsilon = 0.01$. On the other hand, we find that the $X$ asymmetry measure diverges from zero with increasing values of $\Gamma t$, which is consistent with the well known fact that the amplitude damping channel is not symmetric with respect to the representation $\{I, X\}$. Later in this section, we show that it varies with~$\Gamma t$ as $\frac{1}{2} \left( 1 - e^{-\Gamma t} \right)^2$, which is consistent with our simulation results.

Similarly, we test a two-qubit spin-chain system for $\operatorname{SWAP}$, $Z_1Z_2$, and $X_1X_2$ symmetries. We find symmetry to be preserved in the first two cases, wherein the corresponding asymmetry measures are found to be equal to zero. In the case of the $X_1X_2$ symmetry test, however, we find that symmetry is broken. Later in this section, we derive the precise formula according to which the $X_1 X_2$ asymmetry measure is found to depend on $\Gamma$, $t$, and~$J$. Both of the aforementioned examples are discussed in more detail in the subsequent subsections, along with the obtained simulation results and the methods whereby the simulations were performed.

\subsection{Amplitude damping channel}

\begin{figure}
    \centering
    \includegraphics[width=0.95\textwidth]{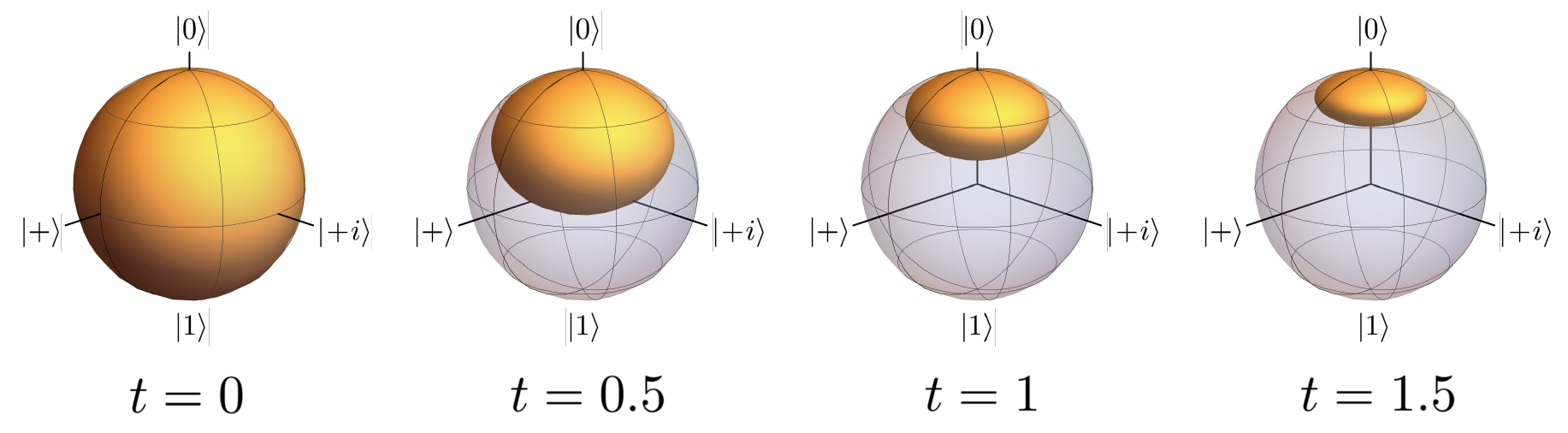}
    \caption{\small Illustration of the action of the $\Gamma=1$ amplitude damping channel on the Bloch sphere over time. All states decay exponentially fast to $|0\rangle$.}
    \label{fig:enter-label}
\end{figure}

\label{sec:amp-damp}

The amplitude damping channel is a quantum channel that models loss of energy from a system to its environment. This can be used to describe open quantum systems that interact with their environment via processes such as spontaneous emission of a single photon from a two-level atomic system.

Continuous-time amplitude damping is generated by a Lindbladian using the raising operator $\sigma^{+} \coloneqq (X+iY)/2$ as a jump operator:
\begin{equation}
    \label{eq:amp-damp-lindbladian}
    \mathcal{L}(\rho) = \Gamma \left(\sigma^+\rho\sigma^- - \frac{1}{2}\left\{\sigma^-\sigma^+, \rho\right\} \right),
\end{equation}
where $\Gamma \geq 0$ represents the rate of $\lvert 1\rangle \rightarrow \lvert 0\rangle$ decay per unit time and $\sigma^- \coloneqq (\sigma^+)^{\dag} = (X -i Y)/2$. We can obtain the superoperator $e^{\mathcal{L}t}$ representing time evolution under this Lindbladian for a time $t$ by mapping Hilbert space operators to Liouville--Fock superoperators under the rule
\begin{equation}
    A \rho B \mapsto (B^{\intercal} \otimes A) \lvert \rho \rangle\!\rangle,
    \label{eq:vectorization}
\end{equation}
where $A$ and $B$ are Hilbert space operators, and $\lvert \rho \rangle\!\rangle$ is the ``vectorized'' version of the density operator, formed by stacking the columns of $\rho$. Applying this to the Lindbladian yields
\begin{align}
    \mathcal{L}(\rho) &= \Gamma \left(\sigma^+\rho\sigma^- - \frac{1}{2}\left\{\sigma^-\sigma^+, \rho\right\}\right) \\
    &\mapsto L |\rho\rangle\!\rangle \coloneqq \Gamma\left(\sigma^+\otimes\sigma^+ - \frac{1}{2}I\otimes (\sigma^-\sigma^+) - \frac{1}{2}(\sigma^-\sigma^+)\otimes I\right) |\rho\rangle\!\rangle,
\end{align}
so that the time evolution superoperator corresponds to
\begin{align}
    e^{L t} &= \exp{
        \begin{pmatrix}
            0 & 0 & 0 & \Gamma t \\
            0 & -\frac{\Gamma t}{2} & 0 & 0 \\
            0 & 0 & -\frac{\Gamma t}{2} & 0 \\
            0 & 0 & 0 & -\Gamma t
        \end{pmatrix}
    }
    =
    \begin{pmatrix}
        1 & 0 & 0 & 1-e^{-\Gamma t} \\
        0 & e^{-\Gamma t / 2} & 0 & 0 \\
        0 & 0 & e^{-\Gamma t / 2} & 0 \\
        0 & 0 & 0 & e^{-\Gamma t}
    \end{pmatrix}.
\end{align}
The action of this matrix on a vectorized density matrix is
\begin{equation}
    e^{Lt}|\rho\rangle\!\rangle = 
    e^{L t}
    \begin{pmatrix}
        \rho_{00} \\ \rho_{10} \\ \rho_{01} \\ \rho_{11}
    \end{pmatrix} =
    \begin{pmatrix}
        \rho_{00} + (1-e^{-\Gamma t})\rho_{11} \\
        e^{-\Gamma t / 2} \rho_{10} \\
        e^{-\Gamma t / 2} \rho_{01} \\
        e^{-\Gamma t} \rho_{11}
    \end{pmatrix}.
\end{equation}
De-vectorizing the above, we find that
\begin{equation}
    e^{\mathcal{L}t}(\rho) = \begin{pmatrix}
        \rho_{00} + (1-e^{-\Gamma t})\rho_{11} & e^{-\Gamma t / 2} \rho_{01} \\
        e^{-\Gamma t / 2} \rho_{10} & e^{-\Gamma t} \rho_{11} \\
    \end{pmatrix}.
    \label{eq:big-gamma-amp-damp}
\end{equation}
It is well known that the time-independent amplitude damping channel $\mathcal{D}_{\gamma}$ for a probability of decay $\gamma$ can be represented by Kraus operators as
\begin{gather}
    K_{0} \coloneqq \begin{pmatrix}
                1 & 0 \\
                0 & \sqrt{1 - \gamma}
            \end{pmatrix}, \qquad
    K_{1} \coloneqq  \begin{pmatrix}
                0 & \sqrt{\gamma} \\
                0 & 0
            \end{pmatrix},
\end{gather}
so that
\begin{equation}
\label{eq:kraus-rep-amp-damp}
    \mathcal{D}_{\gamma}(\rho) = K_{0} \rho K_{0}^{\dag} + K_{1} \rho K_{1}^{\dag} = \begin{pmatrix} \rho_{00} + \gamma\rho_{11} & \sqrt{1 - \gamma}\rho_{01} \\ \sqrt{1 - \gamma}\rho_{10} & (1 - \gamma)\rho_{11} \end{pmatrix}.
\end{equation}
The equivalence of the two representations of the amplitude damping channel in \eqref{eq:big-gamma-amp-damp} and \eqref{eq:kraus-rep-amp-damp} shows that $\gamma = 1-e^{-\Gamma t}$.

\subsubsection{Dependence of \texorpdfstring{$X$}{X} asymmetry measure on \texorpdfstring{$\Gamma t$}{Gamma t}}

\begin{proposition}
\label{prop:X-asym-amp-damp}
    For the amplitude damping channel in \eqref{eq:big-gamma-amp-damp}, the $X$ asymmetry measure defined from \eqref{eq:asymmetry-measure} is given by
    \begin{align}
        a(\mathcal{L}, t, \{ I, X \}) = \frac{1}{2}(1-e^{-\Gamma t})^{2},
    \end{align}
     where $X$ is the $\sigma_{X}$ Pauli matrix.
\end{proposition}

\begin{proof}
Let us consider two channels, denoted by $\mathcal{D}_{\gamma}$ and $\mathcal{X}$. $\mathcal{D}_\gamma$ is the amplitude damping channel, where $\gamma$ denotes the probability of decay.  Its action on a density matrix $\rho$ is defined as in \eqref{eq:kraus-rep-amp-damp}.
The action of $\mathcal{X}$ is defined as $\mathcal{X}(\rho) = X \rho X^{\dag}$. Using these definitions, we calculate the actions of these channels on the elementary matrices ${\{|i \rangle\!\langle j|\}}_{i,j \in \{0, 1\}}$ as follows:
\begin{center}
\def\arraystretch{1.6}
\begin{tabular}{|c|c|}
 \hline
 $\mathcal{D}_\gamma(|0 \rangle\!\langle 0|) = |0 \rangle\!\langle 0|$ & $\mathcal{X}(|0 \rangle\!\langle 0|) = |1 \rangle\!\langle 1|$ \\
 \hline
 $\mathcal{D}_\gamma(|0 \rangle\!\langle 1|) = \sqrt{1-\gamma}|0 \rangle\!\langle 1|$ & $\mathcal{X}(|0 \rangle\!\langle 1|) = |1 \rangle\!\langle 0|$ \\
 \hline
 $\mathcal{D}_\gamma(|1 \rangle\!\langle 0|) = \sqrt{1-\gamma}|1 \rangle\!\langle 0|$ & $\mathcal{X}(|1 \rangle\!\langle 0|) = |0 \rangle\!\langle 1|$ \\
 \hline
 $\mathcal{D}_\gamma(|1 \rangle\!\langle 1|) = \gamma|0 \rangle\!\langle 0| + (1 - \gamma)|1 \rangle\!\langle 1|$ & $\mathcal{X}(|1 \rangle\!\langle 1|) = |0 \rangle\!\langle 0|$ \\
 \hline
\end{tabular}
\end{center}
Next, we define two channels $\{\mathcal{N}, \mathcal{M}\}$, as follows:
\begin{align}
    \mathcal{N} &\coloneqq \mathcal{X} \circ \mathcal{D}_\gamma,\qquad 
    \mathcal{M} \coloneqq \mathcal{D}_\gamma \circ \mathcal{X}.
\end{align}
The actions of the above defined channels with respect to a density matrix $\rho$ are given by $\mathcal{N}(\rho) = \mathcal{X}(\mathcal{D}_\gamma(\rho))$ and $\mathcal{M}(\rho) = \mathcal{D}_\gamma(\mathcal{X}(\rho))$. Again, we may use the above definitions to calculate the following actions:

\vspace{.1in}
\begin{center}
\def\arraystretch{1.6}
\begin{tabular}{|c|c|}
 \hline
 $\mathcal{N}(|0 \rangle\!\langle 0|) = |1 \rangle\!\langle 1|$ & $\mathcal{M}(|0 \rangle\!\langle 0|) = \gamma|0 \rangle\!\langle 0| + (1 - \gamma)|1 \rangle\!\langle 1|$ \\
 \hline
 $\mathcal{N}(|0 \rangle\!\langle 1|) = \sqrt{1-\gamma}|1 \rangle\!\langle 0|$ & $\mathcal{M}(|0 \rangle\!\langle 1|) = \sqrt{1 - \gamma}|1 \rangle\!\langle 0|$ \\
 \hline
 $\mathcal{N}(|1 \rangle\!\langle 0|) = \sqrt{1-\gamma}|0 \rangle\!\langle 1|$ & $\mathcal{M}(|1 \rangle\!\langle 0|) = \sqrt{1-\gamma}|0 \rangle\!\langle 1|$ \\
 \hline
 $\mathcal{N}(|1 \rangle\!\langle 1|) = (1 - \gamma)|0 \rangle\!\langle 0| + \gamma|1 \rangle\!\langle 1|$ & $\mathcal{M}(|1 \rangle\!\langle 1|) = |0 \rangle\!\langle 0|$ \\
 \hline
\end{tabular}
\end{center}
\vspace{.1in}

Let us consider a general unitary representation of the finite discrete group $\mathbb{Z}_2$, given by $\{U(g) \}_{g \in \mathbb{Z}_{2}} = \{ I, W \}$, where $I$ is the two-qubit identity operator, and $W$ is some two-qubit unitary operator satisfying $W^2 = I$. Furthermore, let the unitary channels constituting $\{\mathcal{U}(g)\}_{g \in \mathbb{Z}_{2}}$ and corresponding to $I$ and $W$ be denoted by $\mathcal{I}$ and $\mathcal{W}$ respectively. We may then define the asymmetry measure given in \eqref{eq:asymmetry-measure}, with respect to some Lindbladian channel $e^{\mathcal{L}t}$ and the aforementioned unitary representation $\{ I, W \}$, as follows:
\begin{align}
    &a(\mathcal{L}, t, \{ I, W \}) \notag \\
    &= \frac{1}{2} \sum_{g \in \mathbb{Z}_{2}} \left\Vert \left( \operatorname{id} \otimes \left[ \mathcal{U}(g), e^{\mathcal{L}t} \right] \right) \left( \Phi^{2} \right) \right\Vert_{2}^{2} \\
    &= \frac{1}{2} \left( \left\Vert \left( \operatorname{id} \otimes \left[ \mathcal{I}, e^{\mathcal{L}t} \right] \right) \left( \Phi^{2} \right) \right\Vert_{2}^{2} + \left\Vert \left( \operatorname{id} \otimes \left[ \mathcal{W}, e^{\mathcal{L}t} \right] \right) \left( \Phi^{2} \right) \right\Vert_{2}^{2} \right) \\
    &= \frac{1}{2} \left( \left\Vert \left( \operatorname{id} \otimes \left[ \mathcal{W}, e^{\mathcal{L}t} \right] \right) \left( \Phi^{2} \right) \right\Vert_{2}^{2} \right) \\
    &= \frac{1}{2} \left( \left\Vert \Phi^{\mathcal{W} \circ e^{\mathcal{L}t}} - \Phi^{e^{\mathcal{L}t} \circ \mathcal{W}} \right\Vert_{2}^{2} \right)
    \label{eq:z2-symm-choi-diff}
\end{align}

Now, in order to compute our desired asymmetry measure, we simply substitute the Lindbladian channel $e^{\mathcal{L}t}$ by $\mathcal{D}_{\gamma}$, and the unitary channel $\mathcal{W}$ by $\mathcal{X}$. We then have
\begin{align}
    &a(\mathcal{L}, t, \{ I, X \}) \notag \\ 
    &= \frac{1}{2} \left( \left\Vert \Phi^{\mathcal{X} \circ \mathcal{D}_{\gamma}} - \Phi^{\mathcal{D}_{\gamma} \circ \mathcal{X}} \right\Vert_{2}^{2} \right) \\ 
    &= \frac{1}{2} \left( \left\Vert \Phi^{\mathcal{N}} - \Phi^{\mathcal{M}} \right\Vert_{2}^{2} \right) \\
    &= \frac{1}{2} \left( \operatorname{Tr}\!\left[ \left( \Phi^{\mathcal{N}} \right)^2 \right] + \operatorname{Tr}\!\left[ \left( \Phi^{\mathcal{M}} \right)^2 \right] - 2\operatorname{Tr}\!\left[ \Phi^{\mathcal{N}} \Phi^{\mathcal{M}} \right] \right) \\
    &= \operatorname{Tr} \left[ \left( \Phi^{\mathcal{D}_{\gamma}} \right)^{2} \right] - \operatorname{Tr} \left[ \Phi^{\mathcal{N}} \Phi^{\mathcal{M}} \right] \label{eq:amp-damp-symm-measure-simplified-1},
\end{align}
where the last line follows because $\operatorname{Tr} \left[ \left( \Phi^{\mathcal{D}_{\gamma}} \right)^{2} \right] = \operatorname{Tr} \left[ \left( \Phi^{\mathcal{N}} \right)^{2} \right] = \operatorname{Tr} \left[ \left( \Phi^{\mathcal{M}} \right)^{2} \right]$, which in turn follows from \eqref{eq:simplifying-eqs-ch-meas}. Additionally, from \eqref{eq:Choi-overlap-pf-next}, we know that for any two quantum channels $\mathcal{N}$ and $\mathcal{M}$, the overlap term $\operatorname{Tr}\!\left[ \Phi^{\mathcal{N}} \Phi^{\mathcal{M}} \right]$ may be expressed as
\begin{equation}
\operatorname{Tr}\!\left[ \Phi^{\mathcal{N}} \Phi^{\mathcal{M}} \right] = \frac{1}{d^2}\sum_{i, j}\mathrm{Tr} \left[ \mathcal{N}(|i \rangle\!\langle j|)\mathcal{M}(|j \rangle\!\langle i|) \right].
\end{equation}
Using the above formula, we find that
\vspace{.1in}
\begin{center}
    \def\arraystretch{1.5}
    \begin{tabular}{|c|c|c|}
        \hline
        $|i \rangle\!\langle j|$ & $\mathrm{Tr} \left[ \mathcal{D}_\gamma(|i \rangle\!\langle j|)\mathcal{D}_\gamma(|j \rangle\!\langle i|) \right]$  & $\mathrm{Tr} \left[ \mathcal{N}(|i \rangle\!\langle j|)\mathcal{M}(|j \rangle\!\langle i|) \right]$ \\
        \hline
        $|0 \rangle\!\langle 0|$ & $1$ & $1 - \gamma$ \\
        \hline
        $|0 \rangle\!\langle 1|$ & $1 - \gamma$  & $1 - \gamma$ \\
        \hline
        $|1 \rangle\!\langle 0|$ & $1 - \gamma$  & $1 - \gamma$ \\
        \hline
        $|1 \rangle\!\langle 1|$ & $\gamma^2 + (1 - \gamma)^2$  & $1 - \gamma$ \\
        \hline
    \end{tabular}
\end{center}
\vspace{.1in}
We can now calculate each of the two terms in \eqref{eq:amp-damp-symm-measure-simplified-1}. For $\operatorname{Tr}\!\left[ \left( \Phi^{\mathcal{D}_\gamma} \right)^2 \right]$, we have:
\begin{align}
    \operatorname{Tr}\!\left[ \left( \Phi^{\mathcal{D}_\gamma} \right)^2 \right] &= \frac{1}{d^2}\sum_{i, j}\mathrm{Tr} \left[ \mathcal{D}_\gamma(|i \rangle\!\langle j|)\mathcal{D}_\gamma(|j \rangle\!\langle i|) \right] \\ 
    &= \frac{1}{4}\left[\ 1 + 1 - \gamma + 1 - \gamma + \gamma^2 + (1 - \gamma)^2 \right] \\ 
    &= \frac{1}{2}\left[ \gamma^2 - 2\gamma + 2 \right].
\end{align}
For $\operatorname{Tr}\!\left[ \Phi^{\mathcal{N}} \Phi^{\mathcal{M}} \right]$, we have
\begin{align}
    \operatorname{Tr}\!\left[ \Phi^{\mathcal{N}} \Phi^{\mathcal{M}} \right] &= \frac{1}{d^2}\sum_{i,j}\operatorname{Tr}\!\left[ \mathcal{N}(|i \rangle\!\langle j|)\mathcal{M}(|j \rangle\!\langle i|) \right] \\ 
    &= \frac{1}{4}\left[ 1 - \gamma + 1 - \gamma + 1 - \gamma + 1 - \gamma \right] \\ 
    &= 1 - \gamma.
\end{align}
Plugging the above obtained results into \eqref{eq:amp-damp-symm-measure-simplified-1}, and recalling the identification $\gamma = 1-e^{-\Gamma t}$ made in the previous subsection, we conclude that
\begin{align}
    a(\mathcal{L}, t, \{I, X\}) =  \frac{\gamma^2}{2} = \frac{(1-e^{-\Gamma t})^{2}}{2}.
\end{align}
We have thus computed $X$ asymmetry measure both as a function of the overall probability of decay $\gamma$, as well as the probability of decay per unit time, $\Gamma$, and time $t$.
\end{proof}

\medskip 
We note here that it is interesting to compare the value in Proposition~\ref{prop:X-asym-amp-damp} with Proposition~IV.2 of \cite{leditzky2018}. The latter proposition evaluated an asymmetry measure of the amplitude damping channel in terms of the normalized diamond distance, which is another method for measuring the distance between two quantum channels. Therein, a value of $\frac{1}{2}\left(1-e^{-\Gamma t}\right)$ was reported. Thus, both measures increase with increasing $\Gamma t$, as would be expected for any $X$-asymmetry measure for the amplitude damping channel; however, they increase differently, due to the differing choices of measures. 

\subsubsection{Amplitude damping channel simulation results}

\begin{figure}
    \centering
    \includegraphics[width=0.60\textwidth]{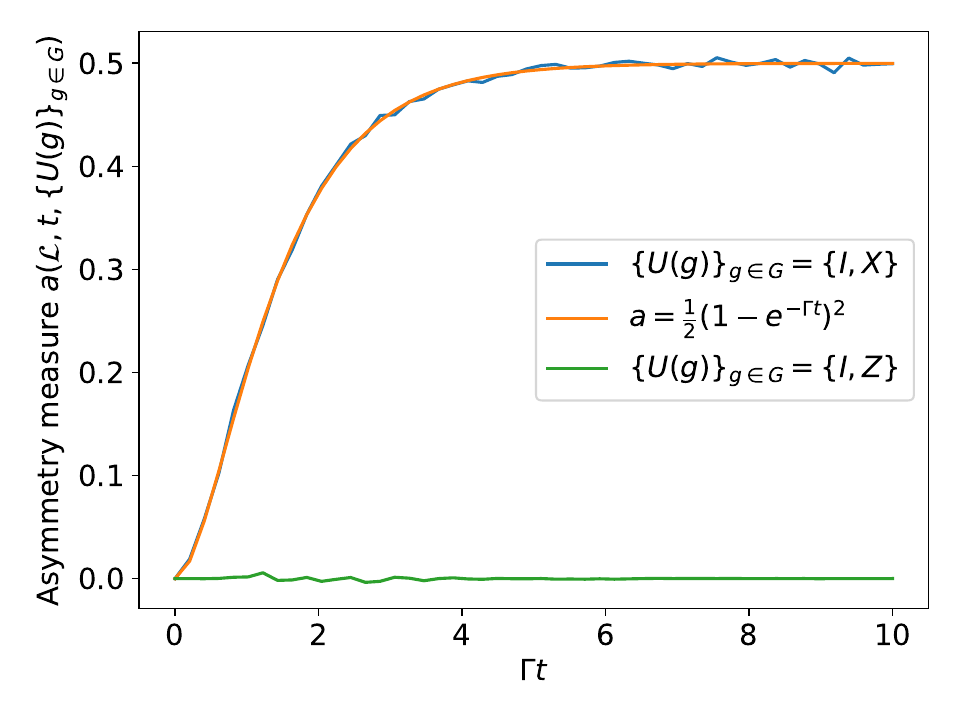}
    \caption{\small 
    Simulated results of applying Algorithm~\ref{alg:HS-estimate} to measure $X$ and $Z$ asymmetries of the single-qubit amplitude damping channel, using Qiskit's noiseless \texttt{QasmSimulator}. The simulation shows that the system maintains $Z$ symmetry for any value of amplitude dissipation $\Gamma t$, while $X$ symmetry is lost for $\Gamma t>0$. The measure of $X$ asymmetry closely matches the analytical expression $\frac{1}{2}\left(1 - e^{-\Gamma t}\right)^2$ in the absence of noise. All simulations were run with a total number of shots determined by the Hoeffding inequality (Theorem~\ref{thm:hoeffding}) with $\epsilon,\delta=0.01$. 
    }
    \label{fig:noiseless-amp-damp}
\end{figure}

We used Qiskit's \texttt{QasmSimulator} to simulate the execution of Algorithm~\ref{alg:HS-estimate} on an idealized quantum processor in order to calculate the $X$ and $Z$ symmetries of the amplitude damping channel. We implement the (non-unitary) amplitude damping channel $\mathcal{D}_{\gamma}$ by means of a unitary extension $D_{\gamma}$, which requires an additional ``environment'' qubit:
\begin{equation}
    \label{eq:amp-damp-unitary}
    D_\gamma = \begin{pmatrix}
        0 & \sqrt{\gamma} & -\sqrt{1-\gamma} & 0 \\
        0 & 0 & 0 & 1 \\
        1 & 0 & 0 & 0 \\
        0 & \sqrt{1-\gamma} & \sqrt{\gamma} & 0
    \end{pmatrix},
\end{equation}
so that
\begin{equation}
    \mathcal{D}_\gamma(\rho) = \Tr_1\!\left[ D_\gamma\left(|0\rangle\!\langle 0| \otimes \rho\right) D_\gamma^\dagger \right].
\end{equation}
Using \eqref{eq:amp-damp-unitary} to implement the amplitude damping channel, we constructed Algorithm~\ref{alg:HS-estimate} in Qiskit and used it to measure the $X$ and $Z$ asymmetries of the channel. We executed the algorithm on Qiskit's \texttt{QasmSimulator}, which emulates an idealized quantum processor with no decoherence. The results, plotted in Figure~\ref{fig:noiseless-amp-damp}, show the expected $Z$ symmetry and $X$ asymmetry, in agreement with the analytical expression.

We also executed the same symmetry tests using Qiskit's \texttt{FakeLima} backend, which provides a depolarizing noise model with parameters estimated from a real quantum processor. These results are plotted in Figure~\ref{fig:noisy-amp-damp}. The $Z$ asymmetry remains zero, while the $X$ asymmetry shows a slight reduction relative to the analytical expression, which accords with the presence of depolarizing noise.

\begin{figure}
    \centering
    \includegraphics[width=0.60\textwidth]{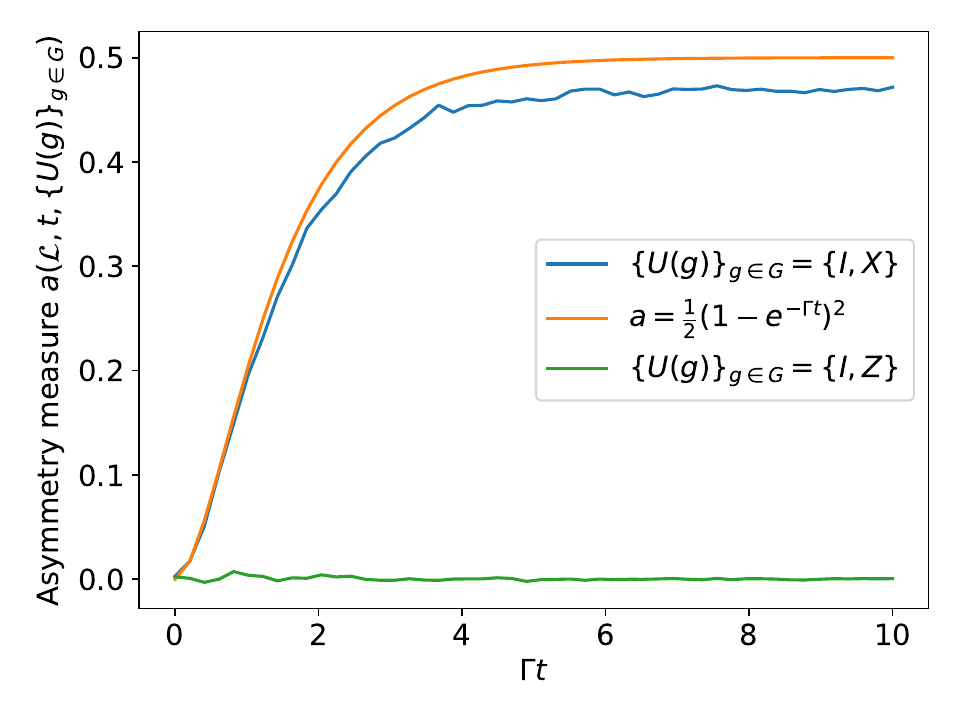}
    \caption{\small 
    Simulated results of applying Algorithm~\ref{alg:HS-estimate} to measure $X$ and $Z$ symmetries of the single-qubit amplitude damping channel, using Qiskit's \texttt{FakeLima} backend. The simulation shows that the system maintains $Z$ symmetry for any value of amplitude dissipation $\Gamma t$, while $X$ symmetry is lost for $\Gamma t>0$. The measure of $X$ asymmetry deviates slightly from the analytical relationship $\frac{1}{2}\left(1 - e^{-\Gamma t}\right)^2$ due to the simulated depolarizing noise. All simulations were run with a total number of shots determined by the Hoeffding inequality (Theorem~\ref{thm:hoeffding}) with $\epsilon,\delta=0.01$.}
    \label{fig:noisy-amp-damp}
\end{figure}

\subsection{\texorpdfstring{$XX$}{XX} Spin chain}

Systems of spin-1/2 particles with nearest-neighbor exchange interactions have been studied for nearly a century and are foundational models in the exploration of magnetism in condensed matter physics \cite{Lieb2009-pt}. See Figure~\ref{fig:spin-chain} for a visualization. In the context of quantum information, spin chains have been studied for potential applications to quantum state transfer. We consider an open $XX$ Heisenberg spin chain consisting of two particles, each of which is subject to amplitude damping dissipation. This system is governed by the Lindblad master equation
\begin{equation}
    \mathcal{L}(\rho) = -i[H, \rho] + \sum_{i=1}^2 \mathcal{L}_i(\rho),
    \label{eq:spin-chain-lindbladian}
\end{equation}
where the Hamiltonian $H$ is given by
\begin{equation}
    H = J (X_1X_2+Y_1Y_2),
\end{equation}
and each $\mathcal{L}_{i}$ term acts on qubit $i$ and is an amplitude dissipation Lindbladian, as defined in \eqref{eq:amp-damp-lindbladian}.
In the above, $J\geq 0$ represents the rate at which excitations hop from one site in the chain to the other.

\begin{figure}
    \centering
    \includegraphics[width=0.85\textwidth]{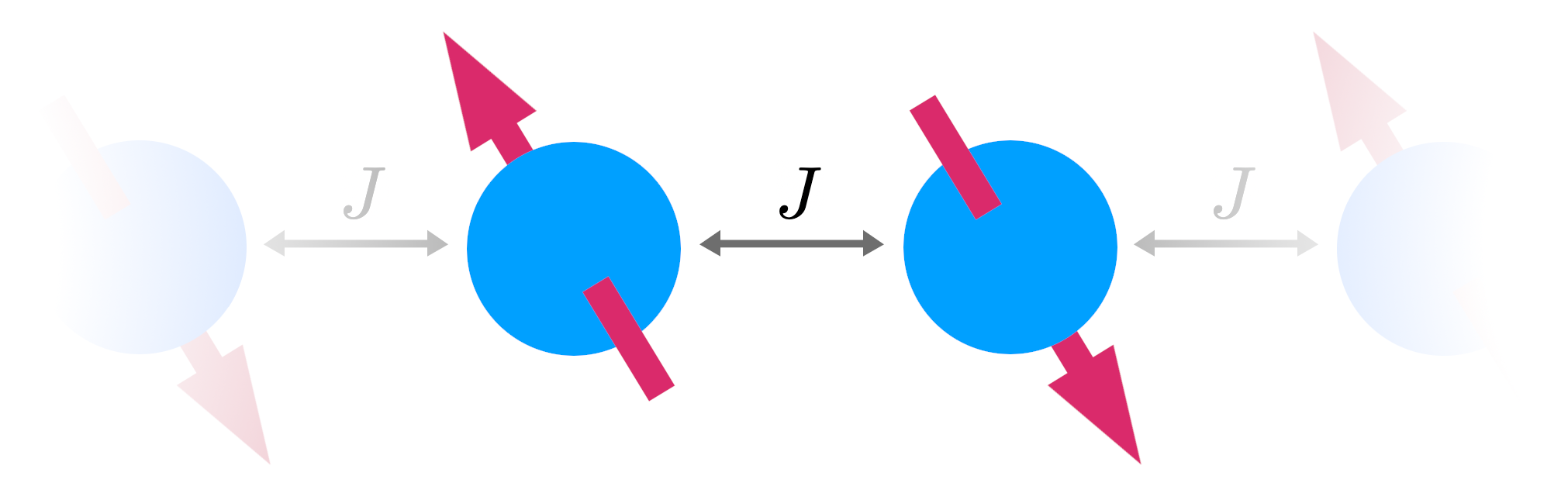}
    \caption{\small Visualization of a spin chain. Each particle's spin couples to that of its neighbors at a rate $J$. We use Algorithm~\ref{alg:HS-estimate} to examine various symmetries of a two-particle spin chain.}
    \label{fig:spin-chain}
\end{figure}

\subsubsection{Spin-chain asymmetries as a function of \texorpdfstring{$\Gamma t$}{Gamma}}

Since amplitude damping dissipation is a longitudinal interaction, this system is manifestly $Z_1Z_2$-symmetric for any amount of damping $\Gamma t$. Conversely, the $X_1X_2$ symmetry of the Hamiltonian is broken by nonzero energy dissipation. Finally, the interactions between the two halves of the system are symmetrical, and so the system is manifestly symmetric under a SWAP of the two particles. Here we use direct calculation of \eqref{eq:asymmetry-measure} to show the $Z_1Z_2$ and SWAP symmetries, and calculate the measure of $X_1X_2$ asymmetry as a function of $\Gamma t$.

\begin{proposition}
    For the open two-qubit $XX$ spin chain defined in \eqref{eq:spin-chain-lindbladian}, the $Z_1Z_2$, $\operatorname{SWAP}$, and $X_1X_2$ asymmetry measures defined from \eqref{eq:asymmetry-measure} are given by
    \begin{equation}
        a(\mathcal{L}, t, \{I,Z_1Z_2\}) = 0, \qquad
        a(\mathcal{L}, t, \{I,\operatorname{SWAP}\}) = 0,
    \end{equation}
    and
    \begin{equation}
        \label{eq:spin-chain-xx-asym}
        \begin{aligned}
        &a(\mathcal{L}, t, \{I,X_1X_2\}) = \\
        &\frac{{e^{-2 t \Gamma} \left( -t^2 \Gamma^2 \cos(4 J) - 16 J^2 \cosh(t \Gamma) + (16 J^2 + t^2 \Gamma^2) \cosh(2 t \Gamma) \right)}}{32 J^2 + 2t^2 \Gamma^2 }.
        \end{aligned}
    \end{equation}
\end{proposition}

\begin{proof}
We calculate the Choi states in \eqref{eq:z2-symm-choi-diff} in terms of the superoperator representations of the channels $\mathcal{W}$ and $e^{\mathcal{L}t}$.
Applying the prescription~\eqref{eq:vectorization} to the terms of the Lindbladian \eqref{eq:spin-chain-lindbladian} and operator $W$, we obtain
\begin{equation}
    W \mapsto W^\intercal \otimes W,
\end{equation}
and
\begin{multline}
    \mathcal{L} \mapsto -i((I\otimes H) - (H^\intercal\otimes I)) \\
    + \sum_i \Gamma\left( \sigma_i^+\otimes\sigma_i^+ - \frac{1}{2}I\otimes (\sigma_i^-\sigma_i^+) - \frac{1}{2}(\sigma_i^-\sigma_i^+)\otimes I\right).
\end{multline}
The latter is straightforwardly exponentiated to obtain a superoperator matrix form of $e^{\mathcal{L}t}$.

Applying these representations to \eqref{eq:choi-state-def}, inserting these Choi states into \eqref{eq:z2-symm-choi-diff} and simplifying with the aid of the computer algebra system Mathematica (code available with our arXiv post), for the three cases of interest $W \in \{Z_1Z_2, \operatorname{SWAP}, X_1X_2\}$ we obtain the expressions given in the proposition.
\end{proof}

\subsubsection{Spin-chain simulation results}

To emulate the dynamics described by \eqref{eq:spin-chain-lindbladian} on a quantum processor as part of the algorithm discussed around \eqref{eq:asymmetry-measure}, we must address two issues. First, the dissipative terms in $\mathcal{L}$ must be replaced by unitary extensions acting on additional ``environment'' qubits, in order to make them implementable by unitary gates. Second, noncommuting terms in $\mathcal{L}$ make it necessary to use Trotterization to implement~$e^{\mathcal{L}t}$.

\begin{figure}
    \centering
    \includegraphics[width=0.60\textwidth]{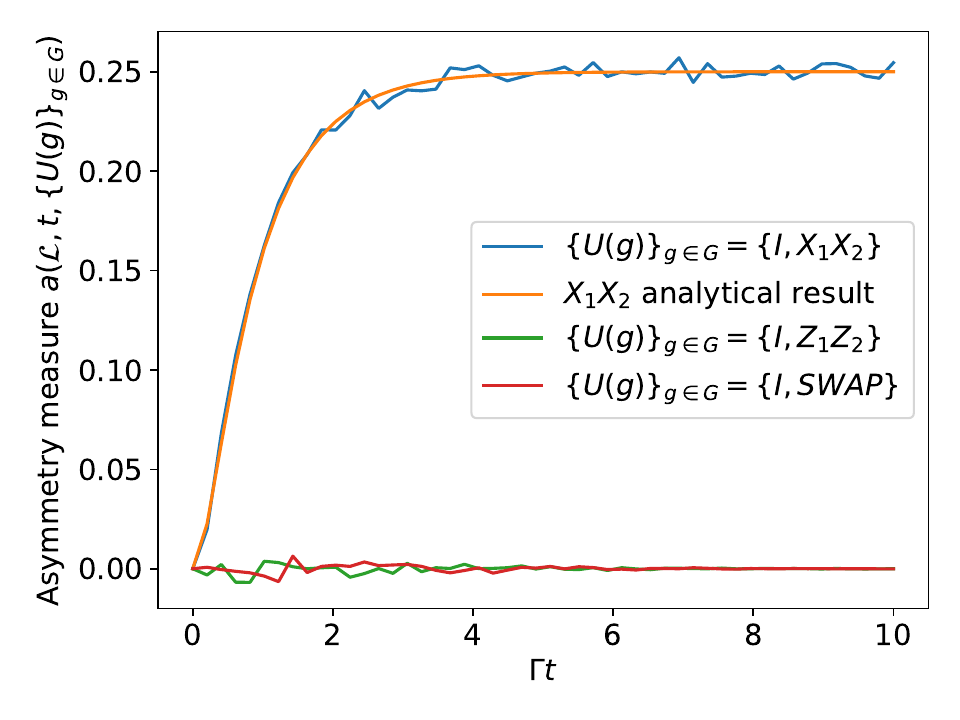}
    \caption{\small Simulation results for applying Algorithm~\ref{alg:HS-estimate} to the two-qubit spin chain using Qiskit's noiseless \texttt{QasmSimulator}. We test the system for $X_1X_2$, $Z_1Z_2$, and $\operatorname{SWAP}$ symmetries. In the latter two cases, the system exhibits symmetry; the asymmetry measure deviates from zero only due to sampling error. In contrast, the system's lack of $X_1X_2$ symmetry is evident for nonzero values of $\Gamma t$, and the value of the asymmetry measure in this case agrees closely with the analytical result from~\eqref{eq:spin-chain-xx-asym}.}
    \label{fig:noiseless-spin-chain}
\end{figure}

We Trotterize the Lindbladian following the prescription in \cite[Proposition~2]{Childs2016EfficientDynamicsArXiv}:
\begin{equation}
    e^{\mathcal{L}t} = \exp\!\left(t\sum_{i=1}^m \mathcal{L}_i\right) \approx \left(\prod_{i=1}^m e^{\mathcal{L}_it/2N} \prod_{j=m}^1 e^{\mathcal{L}_it/2N} \right)^{N},
    \label{eq:lindblad-trotter}
\end{equation}
The specific ordering of terms in this product (forwards and then backwards) results in the first and second orders of the Taylor expansions of the left and right sides of \eqref{eq:lindblad-trotter} to agree exactly.

We note that the terms $ -i[H_j,(\cdot)]$ arising from the Hamiltonian induce unitary evolution, and so can be implemented simply as $e^{-iH_jt}$. The two dissipative Lindblad terms can be implemented using $D_\gamma$, the unitary extension of the amplitude damping channel \eqref{eq:amp-damp-unitary}, provided that the environment qubit is reset to zero before each application of $D_\gamma$. Therefore, each Trotter step of the spin chain Lindbladian can be implemented using
\begin{align}
    \prod_{i=1}^m e^{\mathcal{L}_it/2N} &\mapsto e^{-iX_1X_2 t/2N} e^{-iY_1Y_2 t/2N} D^1_{1-e^{-\Gamma t/2N}} D^2_{1-e^{-\Gamma t/2N}} , \\
    \prod_{i=m}^1 e^{\mathcal{L}_it/2N} &\mapsto D^2_{1-e^{-\Gamma t/2N}} D^1_{1-e^{-\Gamma t/2N}} e^{-iY_1Y_2 t/2N} e^{-iX_1X_2 t/2N}.
\end{align}
This implementation is essentially the same as that presented in Figure~1 of~\cite{Cleve2016EfficientEvolution}, up to a Trotterization of the unitary dynamics, and reordering of the Trotter terms.

\begin{figure}
    \centering
    \includegraphics[width=0.60\textwidth]{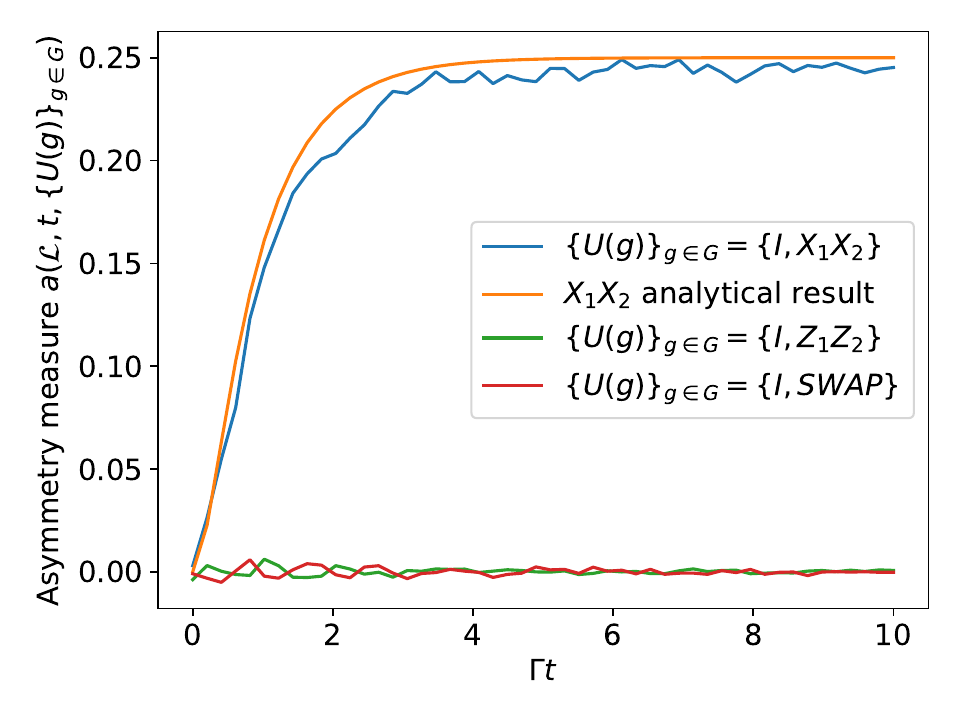}
    \caption{\small Simulation results for applying Algorithm~\ref{alg:HS-estimate} to the two-qubit spin chain using Qiskit's \texttt{FakeLima} backend, which includes realistic depolarizing noise. We test the system for $X_1X_2$, $Z_1Z_2$, and $\operatorname{SWAP}$ symmetries. As in the noiseless simulation (Figure~\ref{fig:noiseless-spin-chain}), we verify $Z_1Z_2$ and $\operatorname{SWAP}$ symmetries, and $X_1X_2$ asymmetry, although the latter deviates from its analytical expression slightly due to the depolarizing noise.}
    \label{fig:noisy-spin-chain}
\end{figure}

We used this formulation to implement Algorithm~\ref{alg:HS-estimate} in Qiskit. We then executed it on Qiskit's \texttt{QasmSimulator} to test the 2-particle spin chain system for $X_1X_2$, $Z_1Z_2$, and $\operatorname{SWAP}$ symmetries, using a number of shots determined by Hoeffding inequality (Theorem~\ref{thm:hoeffding}) with $\epsilon,\delta = 0.01$. The resulting estimates of the asymmetry measure are plotted in Figure~\ref{fig:noiseless-spin-chain}, where we can see that $Z_1Z_2$ and $\operatorname{SWAP}$ symmetries are maintained in the presence of amplitude damping, while $X_1X_2$ symmetry is broken to the degree specified in~\eqref{eq:spin-chain-xx-asym}.

As in the case of our previous simulations of the amplitude damping channel, we also test our spin-chain system for symmetry in the presence of a depolarizing noise model. We do this, as before, by running our code using Qiskit's \texttt{FakeLima} backend. Consistent with the nature of the noise model imported, we find in Figure~\ref{fig:noisy-spin-chain} that the obtained plot of the $X_1 X_2$ asymmetry measure is slightly reduced with respect to the analytical expression as given in \eqref{eq:spin-chain-xx-asym}, while both $Z_1 Z_2$ and $\operatorname{SWAP}$ symmetries appear to be preserved.  

\section{Measurements: Estimating Hilbert--Schmidt distance and testing
symmetries}

\label{sec:meas-ch}

In this section, we consider a special case of the developments in Sections~\ref{sec:HS-Choi-states-algo} and \ref{sec:symm-tests-channels-algos},
when the channels of interest are measurement channels, meaning that they can
be written in the following form:
\begin{equation}
\mathcal{N}(\omega)=\sum_{x}\operatorname{Tr}[N_{x}\omega]|x\rangle\!\langle x|,
\label{eq:povm-channel}
\end{equation}
where $\omega$ is an input state being measured, $\left\{  N_{x}\right\}
_{x}$ is a positive operator-valued measure (POVM) (satisfying $N_{x}\geq0$
for all $x$ and $\sum_{x}N_{x}=I$), and $\left\{  |x\rangle\right\}  _{x}$ is
an orthonormal basis, such that the classical state $|x\rangle\!\langle x|$
encodes the measurement outcome.

We begin by providing an algorithm for estimating the Hilbert--Schmidt
distance of the Choi states of two measurement channels (Section~\ref{sec:meas-ch-HS}). In principle, since
measurement channels are a particular kind of channel, one could simply apply
Algorithm~\ref{alg:HS-estimate} for this task. However, our developments below demonstrate that
this algorithm can be significantly simplified in this case, as a consequence of
the channel outputs being classical.

After that, we then recall the definition of covariance symmetry of
measurement channels and devise an algorithm for testing this symmetry (Section~\ref{sec:meas-sym}). We note that this kind of symmetry is a special case of the channel symmetry mentioned in Remark~\ref{rem:extension-ch-symm}.

\subsection{Estimating the Hilbert--Schmidt distance of the Choi states of
measurement channels}

\label{sec:meas-ch-HS}

We are interested in estimating the Hilbert--Schmidt distance between the Choi states of two measurement channels, defined as in \eqref{eq:meas-channels-N-M} below. Since measurement channels are indeed channels, the expression for the Hilbert--Schmidt distance is precisely the same as that given in \eqref{eq:HS-expand}.

We begin our development with the following lemma, which shows how the various terms in \eqref{eq:HS-expand} simplify when $\mathcal{N}$ and  $\mathcal{M}$ are measurement channels.

\begin{lemma}
\label{lem:Choi-meas-ch-overlap}
Let $\mathcal{N}$ and $\mathcal{M}$ be measurement channels with $d$-dimensional inputs, so that
\begin{equation}
\mathcal{N}(\omega)=\sum_{x}\operatorname{Tr}[N_{x}\omega]|x\rangle\!\langle
x|,\qquad\mathcal{M}(\omega)=\sum_{x}\operatorname{Tr}[M_{x}\omega
]|x\rangle\!\langle x|,
\label{eq:meas-channels-N-M}
\end{equation}
where $\left\{  N_{x}\right\}  _{x}$ and $\left\{  M_{x}\right\}  _{x}$ are
POVMs. Then
\begin{equation}
\operatorname{Tr}[\Phi^{\mathcal{N}}\Phi^{\mathcal{M}}]=\frac{1}{d^{2}}
\sum_{x,y}\delta_{x,y}\operatorname{Tr}[\left(  N_{x}\otimes M_{y}\right)
\left(  \operatorname{SWAP}\right)  ].
\end{equation}

\end{lemma}

\begin{proof}
Recalling \eqref{eq:Choi-overlap-pf-1}--\eqref{eq:Choi-overlap-pf-next}, we find that
\begin{align}
& \operatorname{Tr}[\Phi^{\mathcal{N}}\Phi^{\mathcal{M}}] \notag\\
&  =\frac{1}{d^{2}%
}\sum_{i,j}\operatorname{Tr}[\mathcal{N}(|i\rangle\!\langle j|)\mathcal{M}%
(|j\rangle\!\langle i|)]\\
&  =\frac{1}{d^{2}}\sum_{i,j}\operatorname{Tr}\!\left[  \left(  \sum
_{x}\operatorname{Tr}[N_{x}|i\rangle\!\langle j|]|x\rangle\!\langle x|\right)
\left(  \sum_{y}\operatorname{Tr}[M_{y}|j\rangle\!\langle i|]|y\rangle\!\langle
y|\right)  \right]  \\
&  =\frac{1}{d^{2}}\sum_{i,j,x,y}\operatorname{Tr}[N_{x}|i\rangle\!\langle
j|]\operatorname{Tr}[M_{y}|j\rangle\!\langle i|]\operatorname{Tr}\!\left[
|x\rangle\!\langle x|y\rangle\!\langle y|\right]  \\
&  =\frac{1}{d^{2}}\sum_{i,j,x,y}\delta_{x,y}\operatorname{Tr}[N_{x}%
|i\rangle\!\langle j|]\operatorname{Tr}[M_{y}|j\rangle\!\langle i|]\\
&  =\frac{1}{d^{2}}\sum_{x,y}\delta_{x,y}\operatorname{Tr}[\left(
N_{x}\otimes M_{y}\right)  \left(  \operatorname{SWAP}\right)  ],
\end{align}
concluding the proof.
\end{proof}

\medskip

If the inputs to the channels are $n$-qubit states and the outputs are $m$-bit
strings $\vec{x}$ and $\vec{y}$, then following the development and notation
from \eqref{eq:HS-overlap-qubits}--\eqref{eq:input-bell-state}, we can write
\begin{align}
\operatorname{Tr}[\Phi^{\mathcal{N}}\Phi^{\mathcal{M}}]  &  =\frac{1}{2^{2n}
}\sum_{\vec{x},\vec{y}}\delta_{\vec{x},\vec{y}}\operatorname{Tr}[\left(
N_{\vec{x}}\otimes M_{\vec{y}}\right)  (\text{SWAP}^{(n)})]\\
&  =\frac{1}{2^{2n}}\sum_{\vec{x},\vec{y}\in\left\{  0,1\right\}  ^{m}}
\sum_{\vec{k},\vec{\ell}\in\left\{  0,1\right\}  ^{n}}\delta_{\vec{x},\vec{y}
}\left(  -1\right)  ^{\vec{k}\cdot\vec{\ell}}\operatorname{Tr}[\left(
N_{\vec{x}}\otimes M_{\vec{y}}\right)  (\Phi^{\vec{k}\vec{\ell}})].
\end{align}
Now, by setting $Z\equiv(\vec{X},\vec{Y},\vec{K},\vec{L})$ to be a multi-indexed
random variable taking the value $\delta_{\vec{x},\vec{y}}\left(  -1\right)
^{\vec{k}\cdot\vec{\ell}}$ with probability
\begin{equation}
p(\vec{x},\vec{y},\vec{k},\vec{\ell})    =p(\vec{x},\vec{y}|\vec{k},\vec{\ell
}) \, p(\vec{k},\vec{\ell}),
\end{equation}
where
\begin{align}
     p(\vec{k},\vec{\ell})& =\frac{1}{2^{2n}},\\
p(\vec{x},\vec{y}|\vec{k},\vec{\ell})  &  =\operatorname{Tr}[\left(
N_{\vec{x}}\otimes M_{\vec{y}}\right)  (\Phi^{\vec{k}\vec{\ell}})],
\end{align}
we find from the above that its expectation is given by
\begin{equation}
\mathbb{E}[Z]=\operatorname{Tr}[\Phi^{\mathcal{N}}\Phi^{\mathcal{M}}].
\end{equation}
This leads to the following quantum algorithm for estimating
$\operatorname{Tr}[\Phi^{\mathcal{N}}\Phi^{\mathcal{M}}]$, within additive
error $\varepsilon$ and with success probability not smaller than $1-\delta$,
where $\varepsilon >0 $ and $\delta\in(0,1)$.

\begin{algorithm}
\label{alg:HS-meas-estimate} Given are quantum circuits to implement the
measurement channels $\mathcal{N}$ and~$\mathcal{M}$.

\begin{enumerate}
\item Fix $\varepsilon>0$ and $\delta\in(0,1)$. Set $T\geq\frac{2}
{\varepsilon^{2}}\ln\!\left(  \frac{2}{\delta}\right)  $ and set $t=1$.

\item Generate the bit vectors $\vec{k}$ and $\vec{\ell}$ uniformly at random.

\item Prepare the Bell state $\Phi^{\vec{k}\vec{\ell}}$ on $2n$ qubits (using
the ordering specified in~\eqref{eq:input-bell-state}).

\item Apply the tensor-product measurement channel $\mathcal{N}\otimes
\mathcal{M}$ (using the ordering specified after \eqref{eq:input-bell-state}), which leads to the measurement outcomes $\vec{x}$ and $\vec{y}$.

\item Set $Y_{t}=\delta_{\vec{x},\vec{y}}\left(  -1\right)  ^{\vec{k}\cdot
\vec{\ell}}$.

\item Increment $t$.

\item Repeat Steps 2.-6.~until $t>T$ and then output $\overline{Y}
\coloneqq\frac{1}{T}\sum_{t=1}^{T}Y_{t}$ as an estimate of $\operatorname{Tr}
[\Phi^{\mathcal{N}}\Phi^{\mathcal{M}}]$.
\end{enumerate}
\end{algorithm}

\begin{figure}
\centerline{
\Qcircuit @C=1em @R=0.8em{
|k_1\rangle & & \qw & \gate{H} & \ctrl{3} & \qw & \qw & \multigate{2}{\mathcal{N}}   & \meter & \cw & x_1 \\
|k_2\rangle  & & \qw & \gate{H} & \qw & \ctrl{3} & \qw &  \ghost{\mathcal{N}} &  \meter & \cw & x_2 \\
|k_3\rangle  & & \qw & \gate{H} & \qw & \qw & \ctrl{3} &  \ghost{\mathcal{N}}  & \meter & \cw & x_3 \\
|\ell_1 \rangle & & \qw  & \qw & \targ & \qw  & \qw & \multigate{2}{\mathcal{M}}  & \meter & \cw & y_1 \\
|\ell_2 \rangle & & \qw & \qw & \qw & \targ  & \qw & \ghost{\mathcal{M}} &  \meter & \cw & y_2 \\
|\ell_3 \rangle & & \qw & \qw & \qw & \qw & \targ  & \ghost{\mathcal{M}} &  \meter & \cw & y_3 
}
}
\caption{\small Depiction of the core quantum subroutine given in Steps~2.-4.~of Algorithm~\ref{alg:HS-meas-estimate}, such that the measurement channels $\mathcal{N}$ and $\mathcal{M}$ have three-qubit inputs and three-bit outputs. This algorithm estimates the overlap $\operatorname{Tr}[\Phi^{\mathcal{N}}\Phi^{\mathcal{M}}]$ of the Choi states of the measurement channels.
In this example, the algorithm begins by preparing the classical state $|k_1, k_2, k_3, \ell_1, \ell_2, \ell_3\rangle$, where the values $k_1, k_2, k_3, \ell_1, \ell_2, \ell_3$ are chosen uniformly at random, followed by a sequence of controlled NOTs and Hadamards. Before the measurement channels are applied, the state is thus $|\Phi^{\vec{k} \vec{\ell}}\rangle$, as described in Algorithm~\ref{alg:HS-meas-estimate}. The measurement channels are then applied, leading to the classical bit string $x_1 x_2 x_3 y_1 y_2 y_3$. In the diagram, we depict the realization of the measurement channels $\mathcal{N}$ and $\mathcal{M}$ as black boxes, but in a simulation of them, one might make use of additional environment qubits that are prepared and then discarded.}
\label{fig:Choi-overlap-meas-alg}
\end{figure}

Figure~\ref{fig:Choi-overlap-meas-alg} depicts the core quantum subroutine of Algorithm~\ref{alg:HS-meas-estimate}.
By the Hoeffding inequality (recalled as Theorem~\ref{thm:hoeffding}), we are
guaranteed that the output of Algorithm~\ref{alg:HS-meas-estimate} satisfies
\begin{equation}
\Pr\!\left[  \left\vert \overline{Y}-\operatorname{Tr}[\Phi^{\mathcal{N}}
\Phi^{\mathcal{M}}]\right\vert \leq\varepsilon\right]  \geq1-\delta,
\end{equation}
due to the choice $T\geq\frac{2}{\varepsilon^{2}}\ln\!\left(  \frac{2}{\delta
}\right)  $.

By employing Algorithm~\ref{alg:HS-meas-estimate} three times, we can
thus estimate \eqref{eq:HS-expand} for two measurement channels $\mathcal{N}$ and $\mathcal{M}$ within additive error $\varepsilon$ and
with success probability not smaller than $1-\delta$, by using $O\!\left(
\frac{1}{\varepsilon^{2}}\ln\!\left(  \frac{1}{\delta}\right)  \right)  $
samples of the measurement channels $\mathcal{N}$ and $\mathcal{M}$.

\subsection{Testing symmetries of measurement channels}

\label{sec:meas-sym}

A POVM\ $\{N_{x}\}_{x}$ is covariant  if
there exists a unitary representation $\left\{  U(g)\right\}  _{g\in G}$ of a
group $G$ such that
\begin{equation}
U(g)^{\dag}N_{x}U(g)\in\{N_{x}\}_{x}\quad\forall g\in
G,\,x.\label{eq:def-G-sym-povm}
\end{equation}
Covariant POVMs have been studied previously
\cite{D78,H11book,CdVT03,DJR05}, and they appear in several applications, including state discrimination \cite{KGDdS15} and estimation \cite{CD04}.
Connecting to our previous notion of channel symmetry from Remark~\ref{rem:extension-ch-symm}, a measurement channel
$\mathcal{N}$ is covariant if there exist unitary channel representations
$\left\{  \mathcal{U}(g)\right\}  _{g\in G}$ and $\left\{  \mathcal{W}(g)\right\}  _{g\in G}$ such
that
\begin{equation}
\mathcal{N}\circ\mathcal{U}(g)=\mathcal{W}(g)\circ\mathcal{N}\quad\forall g\in
G.\label{eq:cov-sym-meas-ch}
\end{equation}
Plugging into \eqref{eq:povm-channel}, the condition in
\eqref{eq:cov-sym-meas-ch} becomes
\begin{equation}
\sum_{x}\operatorname{Tr}[U(g)^{\dag}N_{x}U(g)\rho]|x\rangle\!\langle
x|=\sum_{x}\operatorname{Tr}[N_{x}\rho]W(g)|x\rangle\!\langle x|W(g)^{\dag
}\quad\forall g\in G.
\label{eq:meas-cov-symm-cond-1}
\end{equation}
Given that the output system is classical, we can restrict the unitary $W(g)$
to be a shift operator that realizes a permutation $\pi_{g}$ of the classical
letter~$x$, so that
\begin{equation}
W(g)|x\rangle=|\pi_{g}(x)\rangle,
\end{equation}
and thus \eqref{eq:meas-cov-symm-cond-1} becomes
\begin{align}
 \sum_{x}\operatorname{Tr}[U(g)^{\dag}N_{x}U(g)\rho]|x\rangle\!\langle
x|_{X}
&  =\sum_{x}\operatorname{Tr}[N_{x}\rho]|\pi_{g}(x)\rangle\!\langle\pi
_{g}(x)|_{X}\\
&  =\sum_{x}\operatorname{Tr}[N_{\pi_{g}^{-1}(x)}\rho]|x\rangle\!\langle
x|_{X}.
\end{align}
Since this equation holds for every input state $\rho$, we conclude that the
following condition holds for a covariant measurement channel:
\begin{equation}
U(g)^{\dag}N_{x}U(g)=N_{\pi_{g}^{-1}(x)}\quad\forall g\in
G,\ x,\label{eq:cov-povm-conseq}
\end{equation}
coinciding with the definition given in \eqref{eq:def-G-sym-povm}.

We are interested in testing the covariance symmetry of the measurement
channel $\mathcal{N}$, and we can do so by testing the following asymmetry
measure:
\begin{equation}
\frac{1}{\left\vert G\right\vert }\sum_{g\in G}\left\Vert \Phi^{\mathcal{N}
\circ\mathcal{U}(g)}-\Phi^{\mathcal{W}(g)\circ\mathcal{N}}\right\Vert _{2}
^{2},
\label{eq:asymm-meas-for-meas-ch}
\end{equation}
related to the asymmetry measure from \eqref{eq:ch-sym-meas-def}. By invoking Lemma~\ref{lem:ch-sym-meas-reduction}, we find that
\begin{multline}
\frac{1}{\left\vert G\right\vert }\sum_{g\in G}\left\Vert \Phi^{\mathcal{N}
\circ\mathcal{U}(g)}-\Phi^{\mathcal{W}(g)\circ\mathcal{N}}\right\Vert _{2}
^{2}=\frac{2}{d^{2}}\operatorname{Tr}\!\left[  \operatorname{SWAP}\left(
\mathcal{N}\otimes\mathcal{N}\right)  \left(  \operatorname{SWAP}\right)
\right]  \\
-\frac{2}{d^{2}}\operatorname{Tr}\!\left[  \operatorname{SWAP}\left(  \frac
{1}{\left\vert G\right\vert }\sum_{g\in G}\left(  \mathcal{W}(g)\circ
\mathcal{N}\right)  \otimes\left(  \mathcal{N}\circ\mathcal{U}(g)\right)
\right)  \left(  \operatorname{SWAP}\right)  \right]  .
\end{multline}
Now invoking Lemma~\ref{lem:Choi-meas-ch-overlap}, we conclude that
\begin{equation}
\operatorname{Tr}\!\left[  \operatorname{SWAP}\left(  \mathcal{N}\otimes
\mathcal{N}\right)  \left(  \operatorname{SWAP}\right)  \right]  =\sum
_{x,y}\delta_{x,y}\operatorname{Tr}[\left(  N_{x}\otimes N_{y}\right)  \left(
\operatorname{SWAP}\right)  ],
\label{eq:meas-symm-1st-term}
\end{equation}
and
\begin{multline}
\operatorname{Tr}\!\left[  \operatorname{SWAP}\left(  \frac{1}{\left\vert
G\right\vert }\sum_{g\in G}\left(  \mathcal{W}(g)\circ\mathcal{N}\right)
\otimes\left(  \mathcal{N}\circ\mathcal{U}(g)\right)  \right)  \left(
\operatorname{SWAP}\right)  \right]  \\
=\frac{1}{\left\vert G\right\vert }\sum_{g\in G}\sum_{x,y}\delta_{\pi
_{g}(x),y}\operatorname{Tr}\!\left[  \left(  N_{x}\otimes U^{\dag}
(g)N_{y}U(g)\right)  \left(  \operatorname{SWAP}\right)  \right]  .
\label{eq:meas-symm-2nd-term}
\end{multline}
The latter equality follows because $\mathcal{N}\circ\mathcal{U}(g)$ is a
measurement channel with measurement operators $\left\{  U^{\dag}
(g)N_{x}U(g)\right\}  _{x}$ while $\mathcal{W}(g)\circ\mathcal{N}$ is a
measurement channel with measurement operators $\{  N_{\pi_{g}^{-1}(x)}\}  _{x}$. As such, we can employ Algorithm~\ref{alg:HS-meas-estimate} to estimate
both terms in \eqref{eq:meas-symm-1st-term} and \eqref{eq:meas-symm-2nd-term}, and thus estimate~\eqref{eq:asymm-meas-for-meas-ch} by subtracting them and multiplying the result by $\frac{2}{d^2}$. For estimating the latter term, in each step of the algorithm, we pick $g\in G$ uniformly at random, as before.

\section{Conclusion and discussion}

\label{sec:conclusion}

In this work, we proposed asymmetry measures for quantum states, channels, and measurements, as well as efficient quantum algorithms for estimating these measures. A key component of the algorithms for channels and measurements are methods for efficiently estimating the overlap of their Choi states. We demonstrated the channel symmetry testing algorithm in two cases: the single-qubit amplitude damping channel and an open $XX$ spin chain subject to amplitude dissipation. In both cases, we simulated our  algorithm using Qiskit's simulator and found excellent agreement with the analytical expression of the asymmetry measure. 
Finally, we discussed which near-term QPU architectures maximize the system size to be tested using the developed algorithms.



\begin{figure}
    \centering
    \includegraphics[width=\textwidth]{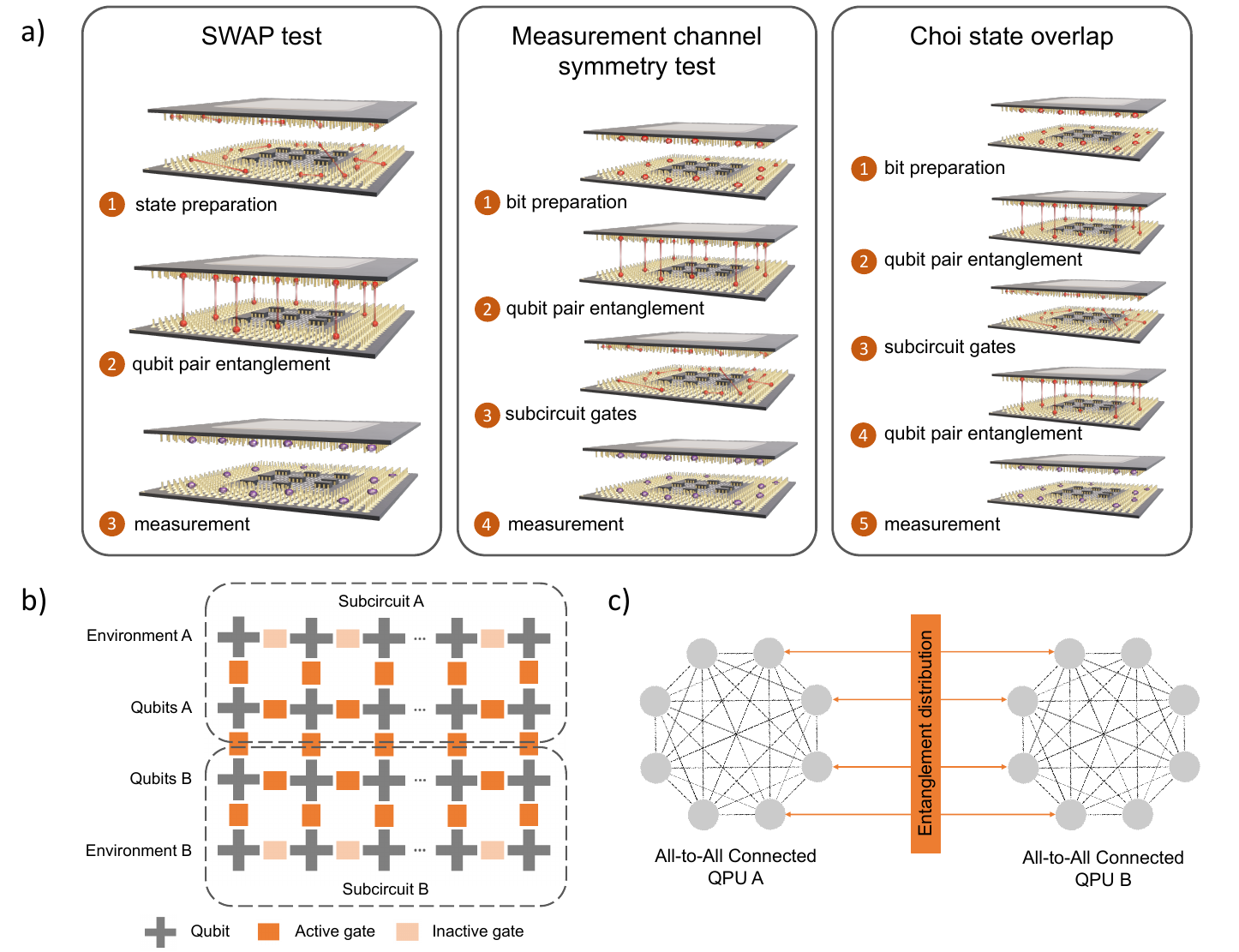}
    \caption{\small Examples of the compatibility of symmetry testing quantum algorithms with quantum processing units (QPUs) of variable connectivity. a) The two subcircuit abstraction of the developed algorithms and their implementation steps. The top chip represents subcircuit A and the bottom the subcircuit B.   b) A two-dimensional array of qubits can explore symmetries in open quantum systems using a one-dimensional chain of nearest-neighbor interactions. c) Two all-to-all connected QPUs with pre-entangled system qubits in an event-ready scheme can explore symmetries of measurements for arbitrary qubit interactions.}
    \label{fig:connectivity}
\end{figure}

\textbf{Prospects for implementing on near-term quantum hardware}---The developed quantum algorithms for symmetry testing can be readily implemented on near-term quantum hardware, as well as potentially guide the development of architectures for upcoming quantum testbeds. We have implemented the Lindbladian symmetry testing algorithm in such a way that the number of physical qubits in hardware is at least four times the number of qubits in the model. The depth of the circuit depends on the selection of Trotterization parameters, and for a specific quantum processing unit (QPU), these parameters should be selected within the hardware coherence limits. 

Furthermore, the qubit connectivity has an important practical role in enabling implementation of the developed algorithms. Each of the algorithms requires the model to be mapped twice to physical qubits in what we will call subcircuits A and B (Figure \ref{fig:connectivity}). Entangling gates are applied to pairs of qubits in subcircuits A and B close to the beginning and/or the end of the algorithm, while the rest of the algorithm requires only local gates inside the subcircuits. This algorithmic split into two computing layers that are cross-connected only once or twice during the implementation of the symmetry testing algorithms lends itself well to upcoming QPU architectures on the IBM Quantum roadmap~\cite{gambetta2022ibm}, Crossbill and Flamingo, for the purposes of maximizing the computable model size. These multi-chip processors are connected either with a smaller number of higher fidelity quantum gates implemented via short chip-to-chip connectors (Crossbill), or a larger number of slower and lower-fidelity quantum gates implemented via long-range couplers (Flamingo). In terms of symmetry testing algorithms where subcircuits A and B would be implemented on different chips, the Crossbill architecture would be suitable for models with a smaller number of qubits and deeper quantum algorithms, while the Flamingo architecture would be more suitable for larger systems that are either implemented via shallower circuits or are executed for algorithms that require only one time-step entanglement via the long-range connectors (SWAP test or measurement channel symmetry test).
 
Some of the existing monolithic quantum processors can be used to efficiently implement symmetry testing of open quantum systems in one-dimensional chain Hamiltonians, which are zoned into subcircuits A and B, as shown in Figure~\ref{fig:connectivity}b. Here, the nearest-neighbor connectivity can be supported by the Google Sycamore superconducting architecture~\cite{arute2019quantum}, while the beyond-the-nearest-neighbor interaction and multi-qubit interactions can be implemented using QuEra Aquila~\cite{wurtz2023aquila} and recent neutral atom quantum hardware advances~\cite{evered2023high}, respectively.

For testing models with higher connectivity, all-to-all connected QPUs, like those offered by IonQ \cite{monroe2021ionq} and Quantinuum \cite{stutz2022trapped} trapped ion hardware or by solid state spin-qubit systems~\cite{bradley2019ten}, can provide more versatility. Since qubits in these systems can generate spin-photon entanglement, multiple QPUs can be connected via photon-mediated entanglement distribution and double the model size in the symmetry testing algorithms (Figure~\ref{fig:connectivity}c). Here, the success of the entanglement distribution is statistical and can be utilized in the event-ready scheme, a frequently employed approach introduced in \cite{yurke1992bell} where photons originating from separate entangling processes in non-local systems become entangled on a beam-splitter and their quantum state projected in a photon-detection process. Obtaining the desired quantum state in the measurement usually takes multiple attempts, and further processing takes place only upon its confirmation when pairs of qubits in separate systems are projected onto desired Bell states. This process is suitable for implementation of the measurement symmetry test (Figure~\ref{fig:connectivity}a) which requires entanglement between subcircuits A and B only at the beginning of the algorithm. To be able to expand this two-QPU implementation from measurement symmetry testing to the state, channel, and Lindbladian symmetry testing, additional work is needed to adapt the protocol to non-deterministic Bell measurements.

\medskip

\textbf{Acknowledgements}---We acknowledge helpful discussions with  Damien Bowen, Soorya Rethinasamy, and Hanna Westerheim. MR acknowledges support by the National Science Foundation (CAREER award No.~2047564) and the Noyce Initiative. MMW acknowledges support from the National Science Foundation under Grant No.~2315398.

\small 

\bibliography{Ref}
\bibliographystyle{plain}

\end{document}